\documentclass{article}
\usepackage{ifpdf}
\ifpdf
\pdfoutput=1
\fi
\usepackage{graphicx}
\usepackage{amsfonts,amssymb,eucal,amsmath}
\usepackage{titlesec}
\usepackage{hyperref}
\usepackage{enumerate}

\newenvironment{proof}{\par \noindent \textsc{Proof.} }{\hfill$\Box$\medskip}
\newtheorem{theorem}{Theorem}
\newtheorem{lemma}[theorem]{Lemma}
\newtheorem{corollary}[theorem]{Corollary}

\begin{document}

\title{Logics for complexity classes}
\author{Vladimir Naidenko\thanks{E-mail: \href{mailto:naidenko@open.by}{\texttt{naidenko@open.by}}}
\\Institute of Mathematics\\National Academy of Sciences of Belarus
\\ul. Surganova 11, Minsk 220012, Belarus
}
\maketitle

\begin{sloppypar}
\begin{abstract}
A new syntactic characterization of problems complete via Turing reductions is presented. General canonical forms are
developed in order to define such problems. One of these forms allows us to define complete problems on ordered
structures, and another form to define them on unordered non-Aristotelian structures. Using the canonical forms, logics
are developed for complete problems in various complexity classes. Evidence is shown that there cannot be any complete
problem on Aristotelian structures for several complexity classes.  Our approach is extended beyond complete problems.
Using a similar form, a logic is developed to capture the complexity class $\mathrm{NP\cap coNP}$ which very likely
contains no complete problem.

\vspace{\baselineskip} \noindent{\it Keywords}: theory of computation, computational complexity, Turing reduction,
completeness, descriptive complexity
\end{abstract}

\section{Introduction}

Since 1974, descriptive complexity characterizes computational complexity in terms of logical languages.
Fagin~\cite{Fagin} first shown that the complexity class $\mathrm{NP}$ coincides with the set of problems expressible
in  second order existential ($\mathrm{SO}\exists$) logic. Stockmeyer~\cite{Stockmeyer} extended Fagin's result to the
polynomial-time hierarchy ($\mathrm{PH}$) characterized by second order logic. Further researches revealed logical
characterizations for various complexity classes~\cite{Immerman,Richerby}. Immerman presented a lot of results on
descriptive complexity in his diagram~\cite{Immerman} whose fragment is shown in Figure~\ref{pic1}.

\begin{figure}[htbp]
\centering
\ifpdf
\includegraphics[width=\textwidth]{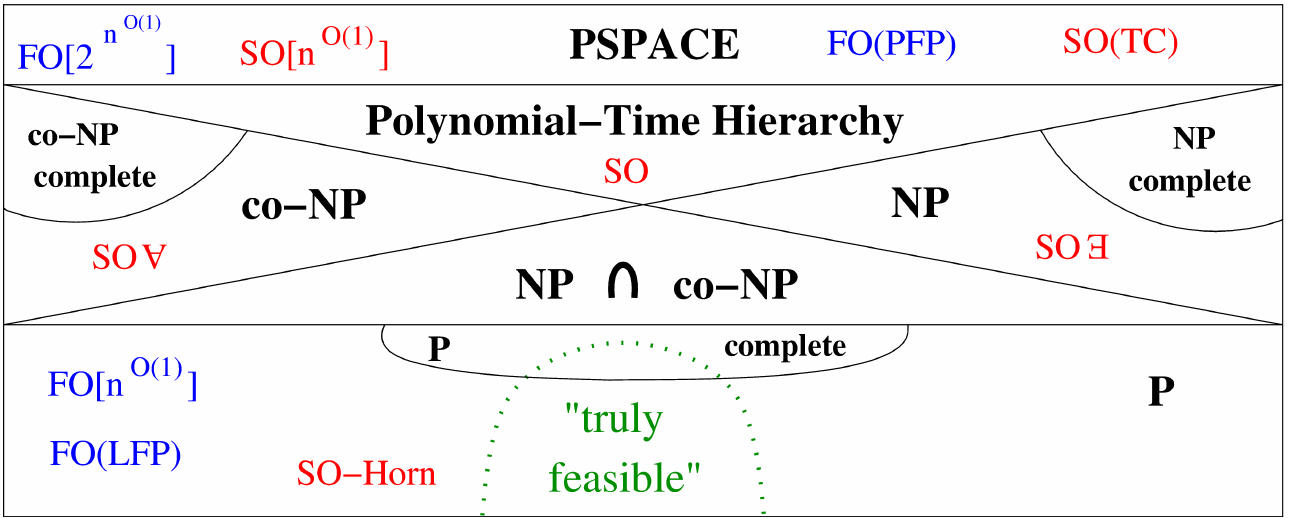}
\else
\includegraphics[width=\textwidth]{pic1.eps}
\fi
\caption[]{The World of Descriptive and Computational Complexity from $\mathrm{P}$ to \mbox{$\mathrm{PSPACE}$}.}
\label{pic1}
\end{figure}

Medina and Immerman~\cite{Medina} addressed the following question. What is it about a $\mathrm{SO}\exists$-sentence
that makes the expressed property $\mathrm{NP}$-comp\-le\-te? They answered this question in part. Namely, a necessary
and sufficient syntactic condition is provided for a \mbox{$\mathrm{SO}\exists$-}se\-n\-ten\-ce to represent
a $\mathrm{NP}$-comp\-le\-te problem on ordered
structures. In addition, the problem is $\mathrm{NP}$-comp\-le\-te via first-order projections (fops). The condition is expressed
in some canonical form based on the $\mathrm{CLIQUE}$ problem. This canonical form provides a syntactic tool for
showing \mbox{$\mathrm{NP}$-}co\-m\-p\-le\-te\-ness: if a \mbox{$\mathrm{SO}\exists$-}se\-n\-ten\-ce is of the form,
then the problem defined by this \mbox{$\mathrm{SO}\exists$-}se\-n\-ten\-ce proves to be
\mbox{$\mathrm{NP}$-}comp\-le\-te via fops.

Bonet and Borges~\cite{Bonet} generalized Medina--Immerman's canonical form in two directions. First, Bonet--Borges's
canonical form works for various complexity classes on ordered structures, and not just $\mathrm{NP}$. Second, the form
can be given in terms of any complete problem in the class, and not just $\mathrm{CLIQUE}$ in case of $\mathrm{NP}$.

However, all the above mentioned forms have a restricted application. No fops are known for the vast majority of
complete problems. For example, in the complexity class $\mathrm{NP}$, only up to fifty
\mbox{$\mathrm{NP}$-}comp\-le\-te problems are known to be complete via fops~\cite{Bonet,Medina}. Moreover, only five
numerical \mbox{$\mathrm{NP}$-}comp\-le\-te problems among thousands of such ones are known to be complete via
fops~\cite{Bonet}. Besides, these forms define complete problems merely on ordered structures, and no canonical form
was known till now to define complete problems on unordered structures.

As shown in Figure~\ref{pic1}, there are complexity classes such as $\mathrm{NP}$-comp\-le\-te problems, $\mathrm{coNP}$-complete problems,
$\mathrm{P}$-complete problems, and $\mathrm{NP\cap coNP}$ for which no logics were known till now. In addition, for the
complexity classes of \mbox{$\mathrm{PSPACE}$-}comp\-le\-te problems and \mbox{$\mathrm{NL}$-}comp\-le\-te problems which are not depicted in
Immerman's diagram, no logics were known as well. The purpose of our research is to develop logics for these classes.

We will consider problems complete via Turing reductions, and not just fops. A new canonical form will be presented to
easily define any complete problem on ordered structures in the following complexity classes: $\mathrm{NL}$,
$\mathrm{P}$, $\mathrm{coNP}$, $\mathrm{NP}$, and $\mathrm{PSPACE}$ (no matter whether it is complete via fops or not).
In that way, we will completely answer the above mentioned question of Medina and Immerman. Moreover, using the form,
logics  will be developed for complete problems on ordered structures in these complexity classes.

We will show an evidence that there cannot be any complete problem on Aristotelian structures in $\mathrm{P}$,
$\mathrm{coNP}$, $\mathrm{NP}$, and $\mathrm{PSPACE}$.

A new canonical form will be presented to define any complete problem on unordered non-Aristotelian structures in the
following complexity classes: $\mathrm{coNP}$, $\mathrm{NP}$, and $\mathrm{PSPACE}$. Using this form, logics will be
developed for complete problems on unordered non-Aristotelian structures in these complexity classes.

Furthermore, we will extend our approach beyond complete problems. By means of a very similar canonic form for defining
problems in the complexity class $\mathrm{NP\cap coNP}$, a logic will be developed to capture $\mathrm{NP\cap coNP}$
which very likely contains no complete problem.

\section{Preliminaries}

We specify some of the notations and conventions used throughout this paper. We denote by $\log x$ the logarithm of $x$
to the base 2. For a real number $x$, we denote
by $\lfloor x\rfloor$ the greatest $n \in \mathbb{Z}$ subject to $n \le x$. We denote by $\widetilde{x}$ the
prefix-free encoding~\cite{Grunwald} of a nonnegative integer~$x$. The length of a string $w$ is denoted by $|w|$. For
strings $w_1$ and $w_2$, we denote by $w_1w_2$ their concatenation. For a string $w$ and a nonnegative integer $x$, we
denote by $w^x$ the concatenation~$\mathop{\underbrace{w\ldots w}}\limits_{x\ \textrm{times}}$ (if $x=0$, then this
concatenation is interpreted as the empty string). By $w[i]$ we denote the $i$-th bit of a nonempty binary string $w$
(for definiteness, we will assume that numbering the bits of any nonempty binary string starts with 1 from left to
right). For a natural number~$x$, we denote by $\ell(x)$ the length of $x$ represented in binary, i.e.
$\ell(x)=\lfloor\log x\rfloor + 1$. By $\ell^{(k)}$ we mean the function
composition~$\mathop{\underbrace{\ell\circ\cdots\circ\ell}}\limits_{k\ \textrm{times}}$~.

We use notations and definitions of finite model theory as stated in \cite{Gradel,Immerman}. For convenience and
without loss of generality, we will consider vocabularies without constant symbols and without function symbols. So, a
vocabulary is a finite set $\tau = \{R^{a_1}_1,\ldots,R^{a_m}_m\}$ of relation symbols of specified arities. A
\mbox{$\tau$-}struc\-tu\-re is a tuple ${A} = (|A|,R^A_1, \ldots , R^A_m)$, where $|A|$ is a nonempty set, called the
universe of $A$, and each $R^{A}_i$ is a relation on $A$ such that $\mathrm{arity}(R^{A}_i) = a_i$, $1 \le i \le m$.  A
$\tau$-struc\-tu\-re is ordered if one of the symbols of its vocabulary $\tau$ is $<$, interpreted as a linear order on
the universe.

A finite $\tau$-struc\-tu\-re is a $\tau$-struc\-tu\-re ${A}$ whose universe $|A|$ is a finite set.
We will assume that the universe of every finite $\tau$-struc\-tu\-re is an initial segment $\{0,\ldots,n-1\}$ of
nonnegative integers, where $n>1$. The notation $\|A\|$ denotes the cardinality of the universe of $A$. Only finite
structures will be considered throughout this paper. Therefore, we will omit the word `finite', when referring to the
structures in what follows.

Let us suppose that $A = (|A|,R^{A}_1 , \ldots,R^{A}_m)$ and $B = (|B|,R^{B}_1 , \ldots , R^{B}_m)$ are two
$\tau$-struc\-tu\-res. An isomorphism between $A$ and $B$ is a mapping $h : |A| \rightarrow |B|$ that satisfies the
following two conditions:

\begin{enumerate}[1)]
  \item $h$ is a bijection.
  \item For every relation symbol $R_i$,\ $1 \le i \le m$, of arity $t$ and for every \mbox{$t$-}tu\-p\-le $(a_1, \ldots ,
      a_t)$ of elements in $|A|$, we have $R^{A}_i (a_1, \ldots , a_t)$ if and only if $R^{B}_i (h(a_1), \ldots ,
      h(a_t))$.
\end{enumerate}

By a model class we mean a set $K$ of $\tau$-struc\-tu\-res of a fixed vocabulary $\tau$ that is closed under
isomorphism, i.e. if $A \in K$ and $A$ is isomorphic to $B$, then $B \in K$ as well. For a vocabulary $\tau$, we by
$\mathrm{STRUC}[\tau]$ denote the model class of all $\tau$-struc\-tu\-res.

Let $\mathcal{L}$ denote a logic. For every vocabulary $\tau$, the language $\mathcal{L}(\tau)$ is the recursive set of
all well-formed sentences (whose elements are called $\mathcal{L}(\tau)$-sen\-ten\-ces) with the symbols of $\tau$ and
with the symbols predefined for the logic $\mathcal{L}$. For example, for first order ($\mathrm{FO}$) logic,
$\mathrm{FO}(\tau)$ is the set of all first order sentences with the symbols of $\tau$, numerical relational symbols:
$=,\ne$, logical connectives: $\wedge, \vee, \neg$, variables: $x,y,z,\ldots$ with or without subscripts, quantifiers:
$\exists, \forall$, and brackets: $(,),[,]$. All the predefined symbols have the standard interpretations. Also, for
$\mathrm{SO}\exists$
logic, $\mathrm{SO}\exists(\tau)$ denotes the set of all second order sentences of the form
$\exists Q_1\ldots\exists Q_l\varphi$, where $\varphi$ is a
{$\mathrm{FO}(\tau\cup\{Q^{a_1}_1,\ldots,Q^{a_l}_l\})$-sen}\-ten\-ce.

In addition, $\models$ is a binary relation between $\mathcal{L}(\tau)$-sen\-ten\-ces and $\tau$-struc\-tu\-res, so
that for each $\mathcal{L}(\tau)$-sen\-ten\-ce $\Gamma$, the set $\{A \in \mathrm{STRUC}[\tau] \mid A \models \Gamma\}$
denoted by $\mathrm{MOD}[\Gamma]$ is closed under isomorphism. Also, we say that a $\mathcal{L}(\tau)$-sen\-ten\-ce
$\Gamma$ defines a model class $K$ if $K = \mathrm{MOD}[\Gamma]$. We call a $\mathcal{L}(\tau)$-sen\-ten\-ce $\Gamma$
logically valid if $\mathrm{STRUC}[\tau]=\mathrm{MOD}[\Gamma]$. We say that two $\mathcal{L}(\tau)$-sen\-ten\-ces
$\Lambda$ and $\Gamma$ are logically equivalent if $\mathrm{MOD}[\Lambda]=\mathrm{MOD}[\Gamma]$.

By $\Gamma(\psi)$ we denote a sentence $\Gamma$ such that $\psi$ occurs in $\Gamma$ as a subformula.
Then, $\Gamma[\psi/\xi]$ denotes the sentence
obtained by replacing each occurrence of $\psi$ by $\xi$ in $\Gamma$.

In order to measure the computational complexity of problems on struc\-tu\-res, we need to represent struc\-tu\-res by
binary strings to be used as inputs for Turing machines. Moreover, simple encodings can be used to represent any
arbitrary objects (integers, sen\-ten\-ces, Turing machines, etc.) as binary strings.

For an object $Z$, we will use $\langle Z\rangle$ to denote some binary representation of~$Z$. For example, if $Z$ is a
natural number, then $\langle Z\rangle$ denotes its usual binary representation. If $Z$ is a sen\-ten\-ce, then
$\langle Z\rangle$ denotes the binary representation of G\"odel number of $Z$. If $Z$ is a Turing machine, then
$\langle Z\rangle$ also denotes the binary representation of G\"odel number of $Z$. If $Z$ is a $\tau$-struc\-tu\-re,
then $\langle Z\rangle$ denotes the binary encoding $\mathrm{bin}(Z)$ as stated in \cite{Immerman}.

We will characterize a model class of $\mathrm{STRUC}[\tau]$ for some fixed vocabulary $\tau$
as a complexity theoretic problem.
Let $\mathcal{L}$ be a logic, $\mathcal{N}$ a
complexity class, and $\tau$ a vocabulary. We say that $\mathcal{L}$ captures $\mathcal{N}$ on $\mathrm{STRUC}[\tau]$
if the following two conditions are satisfied:

\begin{enumerate}[1)]
  \item For every $\mathcal{L}(\tau)$-sen\-ten\-ce $\Gamma$, the model class $\mathrm{MOD}[\Gamma]$
     belongs to $\mathcal{N}$.
  \item For every model class $K \subseteq \mathrm{STRUC}[\tau]$ in $\mathcal{N}$, %
   there exists a $\mathcal{L}(\tau)$-sen\-ten\-ce $\Gamma$ that defines $K$.
\end{enumerate}
Also, we say that $\mathcal{L}$ captures $\mathcal{N}$ on all (ordered) structures if $\mathcal{L}$ captures
$\mathcal{N}$ on $\mathrm{STRUC}[\tau]$, for every vocabulary $\tau$\ (containing $<$).

We need the following definition of oracle Turing machines adapted for struc\-tu\-res \cite{Immerman,Ladner}. An oracle
Turing machine is a Turing machine with a two-way read only input tape, a two-way read-write storage tape, and a
one-way write only oracle tape. Oracle Turing machines have special states: $ACC$, $QUE$, $YES$, and $NO$. The state
$ACC$ is the accepting state. That is, a \mbox{$\tau$-}struc\-tu\-re $A$ is accepted by a Turing machine $T$ if and
only if $T$ on the input $\langle A\rangle$ halts in state $ACC$. The state $QUE$ is the query state. In each state
except $QUE$ the machine may write a symbol onto the oracle tape. In state $QUE$ the machine goes into state $YES$ if
the string written on the oracle tape is a member of the oracle set, otherwise it enters state $NO$. In moving from
state $QUE$ to $YES$ or $NO$ no other action is taken except to erase the oracle tape. By $T^B$ we denote the oracle
Turing machine $T$ equipped with an oracle for the set $B$ of some \mbox{$\tau$-}struc\-tu\-res. For example, the
string $\langle A\rangle$ encoding a \mbox{$\tau$-}struc\-tu\-re $A$ is written on the oracle tape when $T^B$ enters
the query state $QUE$. At the next step, $T^B$ goes into either $YES$ if $A \in B$, or $NO$ otherwise; the oracle tape
is (magically) erased.

We also need the following definitions of Turing machines with bounded resources.

A Turing machine $T$ is polynomial time clocked if the code of $T$ contains a natural number $c$ such that $(n+2)^c$ is
an upper bound for the running time of $T$ on inputs of length $n$.

A Turing machine $T$ is logarithmic space (logspace) bounded if the code of $T$ contains a natural number $c$ such that
$c \log (n+2)$ is an upper bound for the memory space of $T$'s storage tape on inputs of length $n$.

Let we by $\mathcal{N}$ mean some complexity class. For short, we will call an oracle Turing machine
$\mathcal{N}$-spe\-ci\-fic if it is of the same type as the machines that usually realize Turing reductions of problems
in $\mathcal{N}$. So, in case of $\mathcal{N} \in \{\mathrm{coNP}, \mathrm{NP}, \mathrm{PSPACE}\}$, by
$\mathcal{N}$-spe\-ci\-fic Turing machines we mean polynomial time clocked oracle Turing machines. In case of
$\mathcal{N} \in \{\mathrm{NL}, \mathrm{P}\}$, by $\mathcal{N}$-spe\-ci\-fic Turing machines we mean logspace bounded
oracle Turing machines.

Let us recall the following definitions of Turing reducibility and completeness. Given a complexity class
$\mathcal{N}$, for two sets $A$ and $B$ in $\mathcal{N}$, we say that $A$ is Turing reducible to $B$ if there exists a
$\mathcal{N}$-spe\-ci\-fic Turing machine $T^B$ to recognize the set $A$. A set $A \in \mathcal{N}$ is called
$\mathcal{N}$-complete (via Turing reductions) if $X$ is Turing reducible to $A$ for all $X \in \mathcal{N}$.

\section{Aristotelian structures and completeness}

In this section, let $\mathcal{N}$ denote one of the following complexity classes: $\mathrm{P}$, $\mathrm{coNP}$,
$\mathrm{NP}$, or $\mathrm{PSPACE}$. We will present a vocabulary $\tau$ such that the existence of a
\hbox{$\mathcal{N}$-}comp\-le\-te problem defined on the $\tau$-structures is unlikely. Let $\tau=\{R^1_1,\ldots,R^1_m\}$ be a fixed
vocabulary that contains only the symbols of monadic predicates. We call such a vocabulary
Aristotelian~\cite{Lukasiewicz}. Let $A=(|A|,R^A_1,\ldots,R^A_m)$ be an arbitrary \mbox{$\tau$-}struc\-tu\-re. We also
call such a \mbox{$\tau$-}struc\-tu\-re Aristotelian.

Let us consider in detail the binary encoding $\langle A\rangle$ of an Aristotelian \mbox{$\tau$-}struc\-tu\-re $A$. This encoding
$\langle A\rangle$ represents the concatenation of some binary strings: $w^A_1\ldots w^A_m$, where $n=\|A\|$;
$w^A_j\in\{0,1\}^{n}$, $1\le j\le m$. Each string $w^A_j$, $1\le j\le m$, is defined as follows: either $w^A_j[i+1]=1$
if $R^A_j(i)$, or $w^A_j[i+1]=0$ otherwise, for every $0\le i\le n-1$.

Let us show that any Aristotelian $\tau$-structure $A$ can be encoded in a more condensed form as compared with the
encoding $\langle A \rangle$. We write in lexicographic order the sequence $v_1,\ldots,v_{2^{m}}$ of all binary strings
in $\{0,1\}^{m}$. Let us introduce the following characteristic function ${\chi^A:
\{0,\ldots,n-1\}\times\{1,\ldots,2^{m}\}\to \{0,1\}}$ for the $\tau$-structure $A$:
$$ \chi^A(i,j) \triangleq \begin{cases}
1 & \mbox{if } v_j=w^A_1[i+1]w^A_2[i+1]\ldots w^A_m[i+1]\ ;\\
0 & \mbox{otherwise}.\\
\end{cases}
$$
Then, we define a mapping $benc: \mathrm{STRUC}[\tau] \to \{0,1\}^{*}$ as follows:
$$benc(A)\triangleq 1v_1\widetilde{n}_1v_2\widetilde{n}_2\ldots{v_{2^{m}}}\widetilde{n}_{2^{m}}$$
where each $n_j$, $1\le j\le 2^{m}$, is equal to $\sum\limits_{i=0}^{n-1}\chi^A(i,j)$.

Let us estimate the length of $benc(A)$. Note that $|benc(A)|\le 1 + 2^{m}(m+\log n + 2\log\log n + 7)$, where $n=\|A\|$.
Therefore, the length of $benc(A)$ grows logarithmically in $n$ since $m$ is a constant, while the length of $\langle
A\rangle$ grows polynomially in $n$. That is, $benc(A)$ encodes $A$ in a more condensed form as compared with $\langle
A\rangle$.

Note that any Aristotelian $\tau$-structure $A$ can be also encoded as a unary string. We denote by $uenc(A)$ the unary
string  $1^{benc(A)}$, where $benc(A)$ is interpreted here as a number (represented in binary). Then, $uenc(A)$ can be
just considered as the unary encoding of $A$. Since $|uenc(A)|\le2^{|benc(A)|}-1$, we have $|uenc(A)| <
2^{(m+7)2^{m}+1} \cdot(n\log^2n)^{2^{m}}$, where $n=\|A\|$. That is, the length of $uenc(A)$ grows polynomially in $n$
since $m$ is a constant.

We can unambiguously recover $benc(A)$ from $uenc(A)$ as $benc(A)=\langle|uenc(A)|\rangle$ since $benc(A)$ always
starts with 1. However, we cannot unambiguously recover $A$ from $benc(A)$ or $uenc(A)$. Nevertheless, the function
$uenc$ defines the isomorphism property between Aristotelian \mbox{$\tau$-}struc\-tu\-res, which is shown by the
following lemma.

\begin{lemma}
\label{Lemma1} Any two Aristotelian \mbox{$\tau$-}struc\-tu\-res $A$ and $B$ are isomorphic if and only if
$uenc(A)=uenc(B)$. \end{lemma}

\begin{proof} Let $h:|A|\to|B|$ be an isomorphism between the \mbox{$\tau$-}struc\-tu\-res $A$ and $B$. Note that the
statement $\chi^A(a,j)=\chi^B(h(a),j)$ holds for any $a\in|A|$ and $1\le j\le2^{m}$. It follows that $benc(A)=benc(B)$
and, therefore, $uenc(A)=uenc(B)$.

Suppose that $uenc(A)=uenc(B)$. Then, we have $benc(A)=benc(B)$. We will show that the equality $benc(A)=benc(B)$
implies an isomorphism between $A$ and $B$. We denote by $\prec_A$ (resp. $\prec_B$) the following linear order on
$|A|$ (resp. $|B|$). For any $x_1,x_2\in |A|$ (resp. $y_1,y_2\in |B|$), we say that $x_1 \prec_A x_2$ (resp. $y_1
\prec_B y_2$) if there exist natural numbers $j_1$ and $j_2$ such that one of the two conditions is satisfied:
\begin{enumerate}[1)]
  \item $1\le j_1<j_2\le2^m$ and $\chi^A(x_1,j_1)=\chi^A(x_2,j_2)=1$ (resp. $\chi^B(y_1,j_1)=\chi^B(y_2,j_2)=1$).
  \item $1\le j_1=j_2\le2^m$, $\chi^A(x_1,j_1)=\chi^A(x_2,j_2)=1$ (resp. $\chi^B(y_1,j_1)=\chi^B(y_2,j_2)=1$), and
      $x_1<x_2$ (resp. $y_1<y_2$).
\end{enumerate}

Let us write in orders $\prec_A$ and $\prec_B$ the sequences $a_1\prec_A a_2\prec_A\cdots\prec_A a_n$ and $b_1\prec_B
b_2\prec_B\cdots\prec_B b_n$ of all elements in $|A|$ and $|B|$ respectively, where $n=\|A\|$ (note that $\|A\|=\|B\|$
since $benc(A)=benc(B)$). Using these sequences, we define a function $g:|A|\to|B|$ as follows. For any $a\in|A|$, we
have $g(a)\triangleq b_j$, where $j$ such that $a_j=a$.

Note that the equality $\chi^A(a,j)=\chi^B(g(a),j)$ holds for any $a\in|A|$ and $1\le j\le2^{m}$. In other words, for
every $a\in|A|$, we have $R^{A}_i (a)$ if and only if $R^{B}_i (g(a))$, $1 \le i \le m$. That is, this bijection $g$
represents an isomorphism between $A$ and $B$. This concludes the proof of the lemma.
\end{proof}

Then, the following theorem holds.

\begin{theorem}
\label{Theorem1} Let $\tau=\{R^1_1,\ldots,R^1_m\}$ be a fixed Aristotelian vocabulary. If there exists a
\mbox{$\mathcal{N}$-}comp\-le\-te model class of $\tau$-structures, then there exists a unary
\mbox{$\mathcal{N}$-}comp\-le\-te set. \end{theorem}

\begin{proof} Let $L$ be  some $\mathcal{N}$-complete model class of
Aristotelian $\tau$-struc\-tu\-res. Let us consider the set $M=\{u\in1^{*}\mid (\exists A\in
\mathrm{STRUC}[\tau])[u=uenc(A) \wedge A \in L]\}$.

First, we will reduce $L$ to $M$ via a $\mathcal{N}$-specific Turing machine $T^M$. Given a $\tau$-struc\-tu\-re $A$,
this machine $T^M$ transforms the input $\langle A\rangle$ into the binary string $benc(A)$ on the storage tape. Then,
$T^M$ rewrites $benc(A)$ in the unary encoding $uenc(A)$ onto the oracle tape and enters state $QUE$. Next, if
$T^M$ goes into state $YES$ (i.e. $uenc(A)\in M$), then $T^M$ immediately moves to state $ACC$ (i.e. $T^M$ accepts
$A$). Otherwise, $T^M$ rejects $A$. Note that this reduction requires at most logarithmic space since $|benc(A)|$ grows
logarithmically in $\|A\|$. Therefore, $T^M$ is indeed $\mathcal{N}$-specific.

Second, we will show that $M$ belongs to the complexity class $\mathcal{N}$. Let an arbitrary string $u\in1^*$ be
given. Let us decide whether $u\in M$ or not. We transform $u$ into the binary string $w=\langle|u|\rangle$. Then, we
verify in polynomial time whether $w$ is of the form
$1v_1\widetilde{n}_1v_2\widetilde{n}_2\ldots{v_{2^{m}}}\widetilde{n}_{2^{m}}$ or not, where
$\sum\limits_{l=1}^{2^{m}}n_l>1$. If this is not the case, then  $u\not\in M$.
Otherwise, we construct a $\tau$-struc\-tu\-re $A=(|A|,R^{A}_1,\ldots,R^{A}_m)$ in the binary encoding $\langle
A\rangle$ by means of the following algorithm:

\begin{enumerate}[1)]
  \item\label{Step1} Set $i:=1$, $j:=1$, $k:=1$, $n:=\sum\limits_{l=1}^{2^{m}}n_l$.
  \item\label{Step2} Set the binary strings: $w^A_1:=0^n,\ldots,w^A_m:=0^n$.
  \item\label{Step3} If $j>2^{m}$, then stop (the \mbox{$\tau$-}struc\-tu\-re $A$ is formed in the binary encoding
      $\langle A\rangle=w^A_1\ldots w^A_m$). Otherwise, go to step~\ref{Step4}.
  \item\label{Step4} If $n_j=0$, then set $j:=j+1$ and go to step \ref{Step3}. Otherwise, go to step \ref{Step5}.
  \item\label{Step5} For every $1\le r\le m$, set $w^A_r[i]:=1$ if $v_j[r]=1$.
  \item\label{Step6} Set $i:=i+1$ and $k:=k+1$. If $k>n_j$, then set $k:=1$, $j:=j+1$ and go to step \ref{Step3}.
      Otherwise, go to step \ref{Step5}.
\end{enumerate}

Since the set $L$ is closed under isomorphism, the string $u$ belongs to $M$ if and only if this
\mbox{$\tau$-}struc\-tu\-re $A$ is in $L$. It is obvious that constructing the $\tau$-struc\-tu\-re $A$ from the input
$u$ takes polynomial time. Hence, $M$ belongs to $\mathcal{N}$. Thus, the set $M$ is \mbox{$\mathcal{N}$-}comp\-le\-te.
This concludes the proof of the theorem.
\end{proof}

By Theorem~\ref{Theorem1}, it is unlikely that there is a \mbox{$\mathcal{N}$-}comp\-le\-te model class of Aristotelian
structures. For example, the existence of a $\mathrm{NP}$-complete unary language implies $\mathrm{P}=\mathrm{NP}$ (see
\cite{Berman}). Also, the existence of a $\mathrm{P}$-complete sparse language under many-one logspace reduction
implies $\mathrm{L}=\mathrm{P}$ (see \cite{Cai}). Therefore, we will consider \mbox{$\mathcal{N}$-}comp\-le\-te
problems on non-Aristotelian structures in what follows.

\section{Canonical forms for complete problems}

\subsection{Canonical form for complete problems on ordered structures}

We will introduce a canonical form for  $\mathcal{N}$-complete problems on ordered structures. We denote by
$\sigma_{<}$ the vocabulary $\{R^{1},<\}$. We assume that first order logic is extended with the numerical predicate
BIT which is defined in~\cite{Immerman}.  Let $\mathcal{L}_\mathcal{N}$ be a logic capturing $\mathcal{N}$ on
$\mathrm{STRUC}[\sigma_{<}]$, and $\Upsilon(R(x))$ a distinguished \mbox{$\mathcal{L}(\sigma_{<})$-}sen\-ten\-ce
defining some $\mathcal{N}$-complete model class. In this section, we by $\mathcal{N}$ mean one of the following
complexity classes: $\mathrm{NL}$, $\mathrm{P}$, $\mathrm{coNP}$, $\mathrm{NP}$, or $\mathrm{PSPACE}$.
For definiteness, we assume
that $\mathcal{L}_\mathrm{NL}$, $\mathcal{L}_\mathrm{P}$,
$\mathcal{L}_\mathrm{coNP}$, $\mathcal{L}_\mathrm{NP}$, and $\mathcal{L}_\mathrm{PSPACE}$
stand for the following logics:
$\mathrm{FO(TC)}$, $\mathrm{FO(LFP)}$, $\mathrm{SO}\forall$, $\mathrm{SO}\exists$, and $\mathrm{SO(PFP)}$ respectively.

We need to be able to construct sentences defining $\mathcal{N}$-comp\-le\-te problems for various vocabularies
containing $<$, using an operator over $\Upsilon(R(x))$. Let $\tau_{<}=\{ R^{a_1}_1,\ldots,R^{a_m}_m,{<}\}$ be a
fixed arbitrary vocabulary containing $<$. We define the operator $T_{\tau_{<}}$ for mapping $\Upsilon(R(x))$ to a
\mbox{$\mathcal{L}_\mathcal{N}(\tau_{<})$-}sen\-ten\-ce as follows. If $a_1=1$, then
$T_{\tau_{<}}(\Upsilon(R(x)))\triangleq\Upsilon[R(x)/R_1(x)]$. If $a_1>1$, then
$T_{\tau_{<}}(\Upsilon(R(x)))\triangleq\Upsilon[R(x)/\exists y
R_1(\!\!\mathop{\underbrace{y,\ldots,y}}\limits_{a_1-1\ \textrm{times}}\!\!, x)]$, where $y$ is a new variable (we
assume that $y$ does not occur anywhere in $\Upsilon(R(x))$). For short, let  $\Upsilon_{\tau_{<}}$ denote
$T_{\tau_{<}}(\Upsilon(R(x)))$.

Let us show that the \mbox{$\mathcal{L}(\tau_{<})$-}sen\-ten\-ce $\Upsilon_{\tau_{<}}$ defines a
$\mathcal{N}$-comp\-le\-te problem. We will reduce $\mathrm{MOD}[\Upsilon(R(x))]$ to
$\mathrm{MOD}[\Upsilon_{\tau_{<}}]$.
For this purpose, we introduce a \mbox{$\mathcal{N}$-}spe\-ci\-fic Turing machine
$T^{\mathrm{MOD}[\Upsilon_{\tau_{<}}]}$ that recognizes the model class $\mathrm{MOD}[\Upsilon(R(x))]$ as
follows.

Let the binary string $\langle A\rangle$ encoding some \mbox{$\sigma_{<}$-}struc\-tu\-re $A$ be written on the input
tape of $T^{\mathrm{MOD}[\Upsilon_{\tau_{<}}]}$, and $n$ denote $\|A\|$.  At first, the machine
$T^{\mathrm{MOD}[\Upsilon_{\tau_{<}}]}$ computes $n$ as $|\langle A\rangle|$ since the equality $\|A\|=|\langle A\rangle|$
holds in case of the vocabulary $\sigma_{<}$. Then, the machine $T^{\mathrm{MOD}[\Upsilon_{\tau_{<}}]}$ rewrites the
input string $\langle A\rangle$ onto the query tape. Next, the machine $T^{\mathrm{MOD}[\Upsilon_{\tau_{<}}]}$ will
transforms the $\sigma_{<}$-struc\-tu\-re $A$ into a $\tau_{<}$-structure $B$ such that $\|B\|=\|A\|$, either
$R^B_1=\{0\}^{a_1-1}\times R^A$ if $a_1>1$, or $R^B_1=R^A$ if $a_1=1$;
${R^B_2=\varnothing},\ldots,{R^B_m=\varnothing}$. For this purpose, $T^{\mathrm{MOD}[\Upsilon_{\tau_{<}}]}$ appends $N$
zeros to $\langle A\rangle$ on the query tape, obtaining the string $w=\langle A\rangle{0}^{N}$, where $N=n^{a_1} +
n^{a_2} + \ldots + n^{a_m} - n$. Note that $w$ encodes the \mbox{$\tau_{<}$-}struc\-tu\-re $B$ since the initial
substring $\langle A\rangle{0}^{n^{a_1}-n}$ encodes the relation $R^B_1$, and (if any) the other substrings
${0}^{n^{a_2}}$, $\ldots,$ ${0}^{n^{a_m}}$ encode the relations $R^B_2,\ldots,R^B_m$ respectively. Then,
$T^{\mathrm{MOD}[\Upsilon_{\tau_{<}}]}$ enters the query state $QUE$. If $B\models \Upsilon_{\tau_{<}}$, then
$T^{\mathrm{MOD}[\Upsilon_{\tau_{<}}]}$ goes into state $YES$, otherwise it enters state $NO$. At the next step,
$T^{\mathrm{MOD}[\Upsilon_{\tau_{<}}]}$ moves to state $ACC$ from state $YES$ (i.e.
$T^{\mathrm{MOD}[\Upsilon_{\tau_{<}}]}$ accepts $A$). Otherwise, $T^M$ rejects $A$.

It is obvious that $T^{\mathrm{MOD}[\Upsilon_{\tau_{<}}]}$ realizes a logspace reduction from
$\mathrm{MOD}[\Upsilon(R(x))]$ to $\mathrm{MOD}[\Upsilon_{\tau_{<}}]$ since $A$ is a model of
$\Upsilon(R(x))$ if and only if $B$ is a model of $\Upsilon_{\tau_{<}}$. Therefore,
$\mathrm{MOD}[\Upsilon_{\tau_{<}}]$ is $\mathcal{N}$-complete as well as $\mathrm{MOD}[\Upsilon(R(x))]$. Thus,
the \mbox{$\mathcal{L}(\tau_{<})$-}sen\-ten\-ce $\Upsilon_{\tau_{<}}$ indeed defines some $\mathcal{N}$-comp\-le\-te
problem.

With each $\mathcal{L}_\mathcal{N}(\tau_{<})$-sen\-ten\-ce $\Gamma$ and with each $\mathcal{N}$-spe\-ci\-fic Turing
machine $T^\mathrm{MOD[\Gamma]}$, we associate the following set of $\tau_{<}$-struc\-tu\-res:
$$
\begin{array}{rl}
  S^{<}_{\Gamma,T^\mathrm{MOD[\Gamma]}}\triangleq\{\:A\in
\mathrm{STRUC}[\tau_{<}] \mid & \!\!\!\!\! (\forall B \in
\mathrm{STRUC}[\tau_{<}]) \textrm{[}\:\|B\|
> \ell^{(3)}(|\langle A\rangle|)\, \vee\\ & \!\!\!\!\!
    (T^\mathrm{MOD[\Gamma]}\ \textrm{accepts}\ \langle B\rangle \Leftrightarrow B \models \Upsilon_{\tau_{<}})\:\textrm{]}\:\}.
\end{array}
$$

The set $S^{<}_{\Gamma,T^\mathrm{MOD[\Gamma]}}$ has some interesting properties. Note that if the condition
``$T^\mathrm{MOD[\Gamma]}$ accepts $\langle B\rangle$\ $\Leftrightarrow$\ $B \models \Upsilon_{\tau_{<}}$'' is
satisfied for all $B\in \mathrm{STRUC}[\tau_{<}]$, then $T^\mathrm{MOD[\Gamma]}$ is a Turing reduction from
$\mathrm{MOD}[\Upsilon_{\tau_{<}}]$ to $\mathrm{MOD}[\Gamma]$. In this case, $\mathrm{MOD}[\Gamma]$ is a
$\mathcal{N}$-complete model class, and $S^{<}_{\Gamma,T^\mathrm{MOD[\Gamma]}}=\mathrm{STRUC}[\tau_{<}]$. If this is
not the case, then $S^{<}_{\Gamma,T^\mathrm{MOD[\Gamma]}}$ is finite. Let us show how to find a
$\mathrm{FO}(\tau_{<})$-sen\-ten\-ce that defines the model class $S^{<}_{\Gamma,T^\mathrm{MOD[\Gamma]}}$.

The complexity class $\mathrm{DTIME}[\log n]$ consists of all problems decidable in logarithmic time in the input
length $n$. For logarithmic time, an appropriate model of computation is random access machines which can directly
access any memory cell by means of indices. We have the following theorem.

\begin{theorem}
\label{Theorem2} $S^{<}_{\Gamma,T^\mathrm{MOD[\Gamma]}} \in \mathrm{DTIME}[\log n]$. \end{theorem}

\begin{proof}
Let us use a random access machine to decide whether $A\in S^{<}_{\Gamma,T^\mathrm{MOD[\Gamma]}}$ or not, given a
\mbox{$\tau_{<}$-}struc\-tu\-re $A$. As stated in~\cite{Immerman}, in complexity theory, $n$ usually denotes the size
of the input. However, in finite model theory, $n$ denotes the cardinality of the universe. In order to clear up
confusion, let $\hat{n}$ denote $|\langle A\rangle|$.

First, we need to compute the number $\ell^{(3)}(\hat{n})$. For this purpose, we compute $\hat{n}$,
using a binary search~\cite{Aho}. Since we use a random access machine, this can be done in time $O(\log\hat{n})$.
Then, we count the number $\ell(\hat{n})$ of bits in $\hat{n}$ represented in binary. Likewise, we count the number
$\ell^{(2)}(\hat{n})$ of bits in $\ell(\hat{n})$ and then the number $\ell^{(3)}(\hat{n})$ of bits in
$\ell^{(2)}(\hat{n})$. It is obvious that the total time necessary for computing $\ell^{(3)}(\hat{n})$ is $O(\log
\hat{n})$.

Second, we enumerate all \mbox{$\tau_{<}$-}struc\-tu\-res $B\in \mathrm{STRUC}[\tau_{<}]$ such that
$\|B\|\le\ell^{(3)}(\hat{n})$. Since $|\langle B\rangle|\le p(\|B\|)$, the value $|\langle B\rangle|$ is bounded above
by $p(\log \log \log  \hat{n})$, where $p$ is a polynomial dependent on the vocabulary $\tau_{<}$.  Therefore, this
enumeration takes time $O(2^{p(\log \log \log  \hat{n})})$. Then, for each \mbox{$\tau_{<}$-}struc\-tu\-re $B$ of the
enumeration, we need to verify the following condition:
\begin{equation}
T^\mathrm{MOD[\Gamma]}\ \textrm{accepts}\ \langle B\rangle\ \Leftrightarrow\ B \models
\Upsilon_{\tau_{<}}.\label{eq1}
\end{equation}

If condition~\eqref{eq1} is not satisfied for some \mbox{$\tau_{<}$-}struc\-tu\-re $B$ of the enumeration, then the
input \mbox{$\tau_{<}$-}struc\-tu\-re $A$ does not belong to $S^{<}_{\Gamma,T^\mathrm{MOD[\Gamma]}}$. Otherwise, $A \in
S^{<}_{\Gamma,T^\mathrm{MOD[\Gamma]}}$. Now, let us estimate the running time for verifying condition~\eqref{eq1}.

We can decide whether ``$B \models \Upsilon_{\tau_{<}}$'' for one \mbox{$\tau_{<}$-}struc\-tu\-re $B$ in time
$O(2^{p_1(\log \log \log \hat{n})})$, where $p_1$ is a polynomial. This follows from the fact that any model-checking
problem in $\mathcal{N}$ is decidable at most in exponential time. In a similar manner, the verification of
``$T^\mathrm{MOD[\Gamma]}$ accepts $\langle B\rangle$'' takes time $O(2^{p_2(\log \log \log \hat{n})})$, where $p_2$ is
a polynomial. This follows from the fact that a simulation of running $T^\mathrm{MOD[\Gamma]}$ together with a
simulation of oracle queries to $\mathrm{MOD}[\Gamma]$ requires at most exponential time as well. Therefore, the
verification of~\eqref{eq1} takes time $O(2^{p_1(\log \log \log \hat{n})}+2^{p_2(\log \log \log \hat{n})})$, or, in
short, $O(2^{p'(\log \log \log \hat{n})})$, where either $p'=p_1$ in case of
$\displaystyle\lim_{\hat{n}\to\infty}\frac{p_2(\hat{n})}{p_1(\hat{n})}<\infty$, or $p'=p_2$ otherwise.

Consequently, the total time to verify~\eqref{eq1} for all \mbox{$\tau_{<}$-}struc\-tu\-res $B$ of the enumeration is
$O(2^{p(\log \log \log \hat{n})}\cdot2^{p'(\log \log \log  \hat{n})})$, or, in short, $O(2^{p''(\log \log \log
\hat{n})})$, where $p''=p+p'$. Note that $\displaystyle\lim_{\hat{n}\to\infty}\frac{2^{p''(\log \log \log
\hat{n})}}{\log  \hat{n}}=0$ for any polynomial $p''$. Then,
this
time can be approximately estimated at
most as $O(\log \hat{n})$.

Thus, the overall running time consists of the time for computing $\ell^{(3)}(\hat{n})$ and of the time for
verifying~\eqref{eq1} on all \mbox{$\tau_{<}$-}struc\-tu\-res $B$ of the enumeration. These times are both estimated as
$O(\log \hat{n})$. Therefore, so is the overall running time. This concludes the proof of the theorem. \end{proof}

\begin{corollary}
\label{Corollary3} Given a $\mathcal{L}_\mathcal{N}(\tau_{<})$-sen\-ten\-ce $\Gamma$ and a $\mathcal{N}$-spe\-ci\-fic
Turing machine $T^\mathrm{MOD[\Gamma]}$, one can construct a $\mathrm{FO}(\tau_{<})$-sen\-ten\-ce that defines
$S^{<}_{\Gamma,T^\mathrm{MOD[\Gamma]}}$.\end{corollary}

\begin{proof} By Theorem~\ref{Theorem2}, we can use a random access machine $M$ to recognize
$S^{<}_{\Gamma,T^\mathrm{MOD[\Gamma]}}$, where $M$ runs in time $O(\log n)$. Then, one can construct the
\mbox{$\mathrm{FO}(\tau_{<})$-}sen\-ten\-ce from $M$ as shown in \cite{Immerman}, provided that first order logic is
extended with the numerical predicate BIT.
\end{proof}

Then, by $\gamma_{\Gamma,T^\mathrm{MOD[\Gamma]}}$ we denote the constructed $\mathrm{FO}(\tau_{<})$-sen\-ten\-ce
defining the model class $S^{<}_{\Gamma,T^\mathrm{MOD[\Gamma]}}$. This $\gamma_{\Gamma,T^\mathrm{MOD[\Gamma]}}$
characterizes the corresponding Turing machine $T^\mathrm{MOD[\Gamma]}$ in two different ways. First,
$T^\mathrm{MOD[\Gamma]}$ is a Turing reduction from $\mathrm{MOD}[\Upsilon_{\tau_{<}}]$ to $\mathrm{MOD}[\Gamma]$ if
and only if $\gamma_{\Gamma,T^\mathrm{MOD[\Gamma]}}$ is logically valid. Second, $T^\mathrm{MOD[\Gamma]}$ does not
realize any Turing reduction from $\mathrm{MOD}[\Upsilon_{\tau_{<}}]$ to $\mathrm{MOD}[\Gamma]$ if and only if
$\mathrm{MOD}[\gamma_{\Gamma,T^\mathrm{MOD[\Gamma]}}]$ is finite. We call the $\mathrm{FO}(\tau_{<})$-sen\-ten\-ce
$\gamma_{\Gamma,T^\mathrm{MOD[\Gamma]}}$ the characteristic sen\-ten\-ce for the pair
$(\Gamma,T^\mathrm{MOD[\Gamma]})$. The following theorem holds.

\begin{theorem}
\label{Theorem4} Let $\mathcal{L}_\mathcal{N}$ be a logic capturing a complexity class $\mathcal{N}$ on
$\mathrm{STRUC}[\tau_{<}]$, and $\Pi\subseteq \mathrm{STRUC}[\tau_{<}]$ a model class. Then, $\Pi$ is
$\mathcal{N}$-complete if and only if there exists a $\mathcal{L}_\mathcal{N}(\tau_{<})$-sen\-ten\-ce $\Lambda$ such
that $\mathrm{MOD}[\Lambda]=\Pi$ and $\Lambda$ is of the form
\begin{equation}
(\gamma_{\Gamma,T^\mathrm{MOD[\Gamma]}} \wedge \Gamma) \vee
(\neg\gamma_{\Gamma,T^\mathrm{MOD[\Gamma]}} \wedge \Upsilon_{\tau_{<}})\label{eq2}
\end{equation}
where $\Gamma$ is a $\mathcal{L}_\mathcal{N}(\tau_{<})$-sen\-ten\-ce; $T^\mathrm{MOD[\Gamma]}$ a
$\mathcal{N}$-spe\-ci\-fic Turing machine; $\gamma_{\Gamma,T^\mathrm{MOD[\Gamma]}}$ the characteristic sen\-ten\-ce for
$(\Gamma,T^\mathrm{MOD[\Gamma]})$.\end{theorem}

\begin{proof} For short, we denote the form~\eqref{eq2}
by $\Phi_{\Gamma,T^\mathrm{MOD[\Gamma]}}$.

First, we will prove that $\Phi_{\Gamma,T^\mathrm{MOD[\Gamma]}}$ defines a $\mathcal{N}$-complete problem for any pair
$(\Gamma,T^\mathrm{MOD[\Gamma]})$. Let $T^\mathrm{MOD[\Gamma]}$ be a Turing reduction from
$\mathrm{MOD}[\Upsilon_{\tau_{<}}]$ to $\mathrm{MOD}[\Gamma]$. Then, $\gamma_{\Gamma,T^\mathrm{MOD[\Gamma]}}$ is
logically valid. In this case, the $\mathcal{L}_\mathcal{N}(\tau_{<})$-sen\-ten\-ce
$\Phi_{\Gamma,T^\mathrm{MOD[\Gamma]}}$ is logically equivalent to $\Gamma$, and $\Gamma$ defines a
$\mathcal{N}$-complete problem. Therefore, $\Phi_{\Gamma,T^\mathrm{MOD[\Gamma]}}$ defines the same
$\mathcal{N}$-complete problem as well. Now, let $T^\mathrm{MOD[\Gamma]}$ be not a Turing reduction from
$\mathrm{MOD}[\Upsilon_{\tau_{<}}]$ to $\mathrm{MOD}[\Gamma]$. Then,
$\mathrm{MOD}[\gamma_{\Gamma,T^\mathrm{MOD[\Gamma]}}]$ is finite, and
$\mathrm{MOD}[\neg\gamma_{\Gamma,T^\mathrm{MOD[\Gamma]}}]$ is cofinite. Therefore, the model class
$\mathrm{MOD}[\Phi_{\Gamma,T^\mathrm{MOD[\Gamma]}}]$ differs from the model class $\mathrm{MOD}[\Upsilon_{\tau_{<}}]$
only in a finite set of \mbox{$\tau_{<}$-}struc\-tu\-res, i.e. $\mathrm{MOD}[\Phi_{\Gamma,T^\mathrm{MOD[\Gamma]}}]
\setminus \mathrm{MOD}[\Upsilon_{\tau_{<}}]$ and $\mathrm{MOD}[\Upsilon_{\tau_{<}}] \setminus
\mathrm{MOD}[\Phi_{\Gamma,T^\mathrm{MOD[\Gamma]}}]$ are both finite. Hence,
$\mathrm{MOD}[\Phi_{\Gamma,T^\mathrm{MOD[\Gamma]}}]$ is \mbox{$\mathcal{N}$-}comp\-le\-te. Consequently,
$\Phi_{\Gamma,T^\mathrm{MOD[\Gamma]}}$ defines a $\mathcal{N}$-complete problem in this case as well.

Second, we will prove that any $\mathcal{N}$-complete problem can be represented by means of the form
$\Phi_{\Gamma,T^\mathrm{MOD[\Gamma]}}$. Let $\Gamma$ be a $\mathcal{L}_\mathcal{N}(\tau_{<})$-sen\-ten\-ce defining an
arbitrary $\mathcal{N}$-complete problem, and $T^\mathrm{MOD[\Gamma]}$ a $\mathcal{N}$-spe\-ci\-fic Turing machine
realizing a Turing reduction from $\mathrm{MOD}[\Upsilon_{\tau_{<}}]$ to $\mathrm{MOD}[\Gamma]$. Then,
$\gamma_{\Gamma,T^\mathrm{MOD[\Gamma]}}$ is logically valid. In this case, the
$\mathcal{L}_\mathcal{N}(\tau_{<})$-sen\-ten\-ce $\Gamma$ is logically equivalent to the
\mbox{$\mathcal{L}_\mathcal{N}(\tau_{<})$-}sen\-ten\-ce $\Phi_{\Gamma,T^\mathrm{MOD[\Gamma]}}$. Since $\Gamma$ defines
a \mbox{$\mathcal{N}$-}co\-m\-ple\-te problem, so does $\Phi_{\Gamma,T^\mathrm{MOD[\Gamma]}}$. This concludes the proof
of the theorem. \end{proof}

Thus, the form \eqref{eq2} can serve as a canonical form providing a syntactic tool for showing
\mbox{$\mathcal{N}$-}comp\-le\-te\-ness: if a problem on ordered structures in $\mathcal{N}$ is defined by a
sen\-ten\-ce of the form, then the problem proves to be \mbox{$\mathcal{N}$-}comp\-le\-te via Turing reductions.

\subsection{Canonical form for complete problems on unordered structures}

We will introduce a canonical form for $\mathcal{N}$-comp\-le\-te problems on unordered non-Aristotelian structures. We
denote by $\sigma$ the vocabulary $\{R^{2}\}$. In this section, we by $\mathcal{N}$ mean one of the following
complexity classes: $\mathrm{coNP}$, $\mathrm{NP}$, or $\mathrm{PSPACE}$.
 Let
$\mathcal{L}_\mathcal{N}$ be a logic capturing $\mathcal{N}$ on $\mathrm{STRUC}[\sigma]$. For definiteness, we assume
that $\mathcal{L}_\mathrm{coNP}$, $\mathcal{L}_\mathrm{NP}$, and $\mathcal{L}_\mathrm{PSPACE}$
stand for the following logics:
$\mathrm{SO}\forall$, $\mathrm{SO}\exists$, and $\mathrm{SO(PFP)}$ respectively. Let $\Upsilon(R(x,y))$ be a
distinguished $\mathcal{L}(\sigma)$-sen\-ten\-ce defining some $\mathcal{N}$-complete model class.

We need to be able to construct sentences defining $\mathcal{N}$-comp\-le\-te problems for various non-Aristotelian
vocabularies, using an operator over $\Upsilon(R(x,y))$.  Let $\tau=\{ R^{a_1}_1, \ldots,R^{a_m}_m\}$ be a fixed
arbitrary non-Aristotelian vocabulary. We define the operator $T_{\tau}$ for mapping $\Upsilon(R(x,y))$ to a
$\mathcal{L}_\mathcal{N}(\tau)$-sentence as follows. We find the least number $k$ which is subject to ${1\le k\le m}$
and ${a_k>1}$. If $a_k=2$, then $T_\tau(\Upsilon(R(x,y)))\triangleq\Upsilon[R(x,y)/R_k(x,y)]$. If ${a_k>2}$, then
$T_\tau(\Upsilon(R(x,y)))\triangleq\Upsilon[R(x,y) / \exists z
R_k(x,y,\!\!\mathop{\underbrace{z,\ldots,z}}\limits_{a_k-2\ \textrm{times}}\!\!)]$, where $z$ is a new variable (we
assume that $z$ does not occur anywhere in $\Upsilon(R(x,y))$). For short, let $\Upsilon_\tau$ denote
$T_\tau(\Upsilon(R(x,y)))$.

Let us show that the \mbox{$\mathcal{L}(\tau)$-}sen\-ten\-ce $\Upsilon_{\tau}$ defines a $\mathcal{N}$-comp\-le\-te
problem. We can easily reduce $\mathrm{MOD}[\Upsilon(R(x,y))]$ to $\mathrm{MOD}[\Upsilon_\tau]$ in the following way.
Let $A$ be an arbitrary $\sigma$-structure. Then, we take a $\tau$-structure $B$ such that $\|B\|=\|A\|$, either
$R_k^{B}=R^A$ if $a_k=2$, or $R_k^{B}=R^A \times \{0\}^{a_k-2}$ if $a_k>2$; ${R^{B}_1=\varnothing}, \ldots,
{R^{B}_{k-1}=\varnothing}, {R^{B}_{k+1}=\varnothing}, \ldots, {R^{B}_m=\varnothing}$. Note that $A$ is a model of
$\Upsilon(R(x,y))$ if and only if $B$ is a model of $\Upsilon_\tau$. It is obvious that $\mathrm{MOD}[\Upsilon_\tau]$
is $\mathcal{N}$-complete as well as $\mathrm{MOD}[\Upsilon(R(x,y))]$. Thus, the
\mbox{$\mathcal{L}(\tau)$-}sen\-ten\-ce $\Upsilon_{\tau}$ indeed defines some $\mathcal{N}$-comp\-le\-te problem.

With each $\mathcal{L}_\mathcal{N}(\tau)$-sen\-ten\-ce $\Gamma$ and with each $\mathcal{N}$-spe\-ci\-fic Turing machine
$T^\mathrm{MOD[\Gamma]}$, we associate the following set of $\tau$-struc\-tu\-res:
$$
\begin{array}{rl}
  S_{\Gamma,T^\mathrm{MOD[\Gamma]}}\triangleq\{\:A\in
\mathrm{STRUC}[\tau] \mid & \!\!\!\!\! (\forall B \in
\mathrm{STRUC}[\tau]) \textrm{[}\:\|B\|
> \ell^{(2)}(|\langle A\rangle|)\, \vee\\ & \!\!\!\!\!
    (T^\mathrm{MOD[\Gamma]}\ \textrm{accepts}\ \langle B\rangle \Leftrightarrow B \models \Upsilon_\tau)\:\textrm{]}\:\}.
\end{array}
$$

Like $S^{<}_{\Gamma,T^\mathrm{MOD[\Gamma]}}$, the set $S_{\Gamma,T^\mathrm{MOD[\Gamma]}}$ has similar properties. Note
that if the condition ``$T^\mathrm{MOD[\Gamma]}$ accepts $\langle B\rangle$\ $\Leftrightarrow$\ $B \models
\Upsilon_\tau$'' is satisfied for all $B\in \mathrm{STRUC}[\tau]$, then $T^\mathrm{MOD[\Gamma]}$ is a Turing reduction
from $\mathrm{MOD}[\Upsilon_\tau]$ to $\mathrm{MOD}[\Gamma]$. In this case, $\mathrm{MOD}[\Gamma]$ is a
$\mathcal{N}$-complete model class, and $S_{\Gamma,T^\mathrm{MOD[\Gamma]}}=\mathrm{STRUC}[\tau]$. If this is not the
case, then $S_{\Gamma,T^\mathrm{MOD[\Gamma]}}$ is finite. Let us show how to find a
$\mathcal{L}_\mathcal{N}(\tau)$-sen\-ten\-ce that defines the model class $S_{\Gamma,T^\mathrm{MOD[\Gamma]}}$.

The complexity class $\mathrm{DTIME}[n]$ consists of all problems decidable in linear time in the input size $n$. The
following theorem holds.

\begin{theorem}
\label{Theorem5} $S_{\Gamma,T^\mathrm{MOD[\Gamma]}} \in \mathrm{DTIME}[n]$.\end{theorem}

\begin{proof} It is similar to the proof of Theorem~\ref{Theorem2} in many respects.
Let us decide in linear time whether $A\in S_{\Gamma,T^\mathrm{MOD[\Gamma]}}$ or not, given a
\mbox{$\tau$-}struc\-tu\-re $A$. Let $\hat{n}$ denote~$|\langle A\rangle|$.

First, we need to compute the number $\ell^{(2)}(\hat{n})$. Suppose that we use a usual Turing machine rather than a
random access machine. In this case, it is obvious that the time necessary for computing $\ell^{(2)}(\hat{n})$ is
$O(\hat{n})$.

Second, we enumerate all \mbox{$\tau$-}struc\-tu\-res $B\in \mathrm{STRUC}[\tau]$ such that
$\|B\|\le\ell^{(2)}(\hat{n})$. Since $|\langle B\rangle|\le p(\|B\|)$, the value $|\langle B\rangle|$ is bounded above
by $p(\log \log  \hat{n})$, where $p$ is a polynomial dependent on the vocabulary $\tau$.  Therefore, this enumeration
takes time $O(2^{p(\log \log \hat{n})})$. Then, for each \mbox{$\tau$-}struc\-tu\-re $B$ of the enumeration, we need to
verify the following condition:
\begin{equation}
T^\mathrm{MOD[\Gamma]}\ \textrm{accepts}\ \langle B\rangle\ \Leftrightarrow\ B \models
\Upsilon_\tau.\label{eq3}
\end{equation}

If condition~\eqref{eq3} is not satisfied for some \mbox{$\tau$-}struc\-tu\-re $B$ of the enumeration, then the input
\mbox{$\tau$-}struc\-tu\-re $A$ does not belong to $S_{\Gamma,T^\mathrm{MOD[\Gamma]}}$. Otherwise, $A \in
S_{\Gamma,T^\mathrm{MOD[\Gamma]}}$. Now, let us estimate the running time for verifying condition~\eqref{eq3}.

We can decide whether ``$B \models \Upsilon_\tau$'' for one \mbox{$\tau$-}struc\-tu\-re $B$ in time $O(2^{p_1(\log \log
\hat{n})})$, where $p_1$ is a polynomial. This follows from the fact that any model-checking problem in $\mathcal{N}$
is decidable at most in exponential time. In a similar manner, the verification of ``$T^\mathrm{MOD[\Gamma]}$ accepts
$\langle B\rangle$'' takes time $O(2^{p_2(\log \log \hat{n})})$, where $p_2$ is a polynomial. This follows from the
fact that a simulation of running $T^\mathrm{MOD[\Gamma]}$ together with a simulation of oracle queries to
$\mathrm{MOD}[\Gamma]$ requires at most exponential time as well. Therefore, the verification of~\eqref{eq3} takes time
$O(2^{p_1(\log \log \hat{n})}+2^{p_2(\log \log \hat{n})})$, or, in short, $O(2^{p'(\log \log \hat{n})})$, where either
$p'=p_1$ in case of $\displaystyle\lim_{\hat{n}\to\infty}\frac{p_2(\hat{n})}{p_1(\hat{n})}<\infty$, or $p'=p_2$
otherwise.

Consequently, the total time to verify~\eqref{eq3} for all \mbox{$\tau$-}struc\-tu\-res $B$ of the enumeration is
$O(2^{p(\log \log \hat{n})}\cdot2^{p'(\log \log  \hat{n})})$, or, in short, $O(2^{p''(\log \log  \hat{n})})$, where
$p''=p+p'$. Note that $\displaystyle\lim_{\hat{n}\to\infty}\frac{2^{p''(\log \log \hat{n})}}{\hat{n}}=0$ for any
polynomial $p''$. Then,
this
time can be approximately estimated at most as $O(\hat{n})$.

Thus, the overall running time consists of the time for computing $\ell^{(2)}(\hat{n})$ and of the time for
verifying~\eqref{eq3} on all \mbox{$\tau$-}struc\-tu\-res $B$ of the enumeration. The both of these times are estimated
as $O(\hat{n})$. Therefore, so is the overall running time. This concludes the proof of the theorem. \end{proof}

\begin{corollary}
\label{Corollary6} Given a $\mathcal{L}_\mathcal{N}(\tau)$-sen\-ten\-ce $\Gamma$ and a $\mathcal{N}$-spe\-ci\-fic
Turing machine $T^\mathrm{MOD[\Gamma]}$, one can construct a $\mathcal{L}_\mathcal{N}(\tau)$-sen\-ten\-ce that defines
$S_{\Gamma,T^\mathrm{MOD[\Gamma]}}$.\end{corollary}

\begin{proof} By Fagin's theorem~\cite{Fagin}, in case of $\mathcal{N}=\mathrm{NP}$, we can
construct a \mbox{$\mathrm{SO}\exists(\tau)$-se}\-n\-ten\-ce to define the model class
$S_{\Gamma,T^\mathrm{MOD[\Gamma]}}$ since $S_{\Gamma,T^\mathrm{MOD[\Gamma]}}\in\mathrm{DTIME}[n]$. In case of
$\mathcal{N}=\mathrm{PSPACE}$, we can use the same $\mathrm{SO}\exists(\tau)$-sen\-ten\-ce since $\mathrm{SO}\exists$
is a fragment of $\mathrm{SO(PFP)}$. In case of $\mathcal{N}=\mathrm{coNP}$, we can at first construct a
$\mathrm{SO}\exists(\tau)$-sen\-ten\-ce to define the model class $\mathrm{STRUC}[\tau]\setminus
S_{\Gamma,T^\mathrm{MOD[\Gamma]}}$ since $\mathrm{STRUC}[\tau]\setminus
S_{\Gamma,T^\mathrm{MOD[\Gamma]}}\in\mathrm{DTIME}[n]$ as well as
$S_{\Gamma,T^\mathrm{MOD[\Gamma]}}\in\mathrm{DTIME}[n]$. Then, the $\mathrm{SO}\forall(\tau)$-sen\-ten\-ce defining
$S_{\Gamma,T^\mathrm{MOD[\Gamma]}}$ is obtained by logical negation of this $\mathrm{SO}\exists(\tau)$-sen\-ten\-ce.
\end{proof}

Then, by $\Theta_{\Gamma,T^\mathrm{MOD[\Gamma]}}$ we denote the constructed
$\mathcal{L}_\mathcal{N}(\tau)$-sen\-ten\-ce defining the model class $S_{\Gamma,T^\mathrm{MOD[\Gamma]}}$. This
$\Theta_{\Gamma,T^\mathrm{MOD[\Gamma]}}$ characterizes the corresponding Turing machine $T^\mathrm{MOD[\Gamma]}$ in two
different ways. First, $T^\mathrm{MOD[\Gamma]}$ is a Turing reduction from $\mathrm{MOD}[\Upsilon_\tau]$ to
$\mathrm{MOD}[\Gamma]$ if and only if $\Theta_{\Gamma,T^\mathrm{MOD[\Gamma]}}$ is logically valid. Second,
$T^\mathrm{MOD[\Gamma]}$ does not realize any Turing reduction from $\mathrm{MOD}[\Upsilon_\tau]$ to
$\mathrm{MOD}[\Gamma]$ if and only if $\mathrm{MOD}[\Theta_{\Gamma,T^\mathrm{MOD[\Gamma]}}]$ is finite. We call the
\mbox{$\mathcal{L}_\mathcal{N}(\tau)$-}sen\-ten\-ce $\Theta_{\Gamma,T^\mathrm{MOD[\Gamma]}}$ the characteristic
sen\-ten\-ce for the pair $(\Gamma,T^\mathrm{MOD[\Gamma]})$.

Similarly, by $\overline{\Theta}_{\Gamma,T^\mathrm{MOD[\Gamma]}}$ we denote a
$\mathcal{L}_\mathcal{N}(\tau)$-sen\-ten\-ce that defines the model class $\mathrm{STRUC}[\tau]\setminus
S_{\Gamma,T^\mathrm{MOD[\Gamma]}}$. Given a $\mathcal{L}_\mathcal{N}(\tau)$-sen\-ten\-ce $\Gamma$ and a
$\mathcal{N}$-spe\-ci\-fic Turing machine $T^\mathrm{MOD[\Gamma]}$, this
$\overline{\Theta}_{\Gamma,T^\mathrm{MOD[\Gamma]}}$ can be also constructed since $\mathrm{STRUC}[\tau]\setminus
S_{\Gamma,T^\mathrm{MOD[\Gamma]}}\in\mathrm{DTIME}[n]$ as well as
$S_{\Gamma,T^\mathrm{MOD[\Gamma]}}\in\mathrm{DTIME}[n]$. Note that $\overline{\Theta}_{\Gamma,T^\mathrm{MOD[\Gamma]}}$
is logically equivalent to $\neg{\Theta}_{\Gamma,T^\mathrm{MOD[\Gamma]}}$. However, we cannot directly use
$\neg{\Theta}_{\Gamma,T^\mathrm{MOD[\Gamma]}}$ to construct a canonical form since
$\neg{\Theta}_{\Gamma,T^\mathrm{MOD[\Gamma]}}$ is not a \mbox{$\mathcal{L}_\mathcal{N}(\tau)$-}sen\-ten\-ce in case of
$\mathcal{N}\in\{\mathrm{NP}, \mathrm{coNP}\}$. Therefore, we will use
$\overline{\Theta}_{\Gamma,T^\mathrm{MOD[\Gamma]}}$ instead. We call the
\mbox{$\mathcal{L}_\mathcal{N}(\tau)$-}sen\-ten\-ce $\overline{\Theta}_{\Gamma,T^\mathrm{MOD[\Gamma]}}$ the
complementary sentence for ${\Theta}_{\Gamma,T^\mathrm{MOD[\Gamma]}}$. Then, the following theorem holds.

\begin{theorem}
\label{Theorem7} Let $\mathcal{L}_\mathcal{N}$ be a logic capturing a complexity class
$\mathcal{N}$ on $\mathrm{STRUC}[\tau]$,
and $\Pi\subseteq \mathrm{STRUC}[\tau]$ a model class. Then, $\Pi$ is $\mathcal{N}$-complete if and only if there
exists a $\mathcal{L}_\mathcal{N}(\tau)$-sen\-ten\-ce $\Lambda$ such that $\mathrm{MOD}[\Lambda]=\Pi$ and $\Lambda$ is
of the form
\begin{equation}
(\Theta_{\Gamma,T^\mathrm{MOD[\Gamma]}} \wedge \Gamma) \vee
(\overline\Theta_{\Gamma,T^\mathrm{MOD[\Gamma]}} \wedge \Upsilon_\tau)\label{eq4}
\end{equation}
where $\Gamma$ is a $\mathcal{L}_\mathcal{N}(\tau)$-sen\-ten\-ce; $T^\mathrm{MOD[\Gamma]}$ a $\mathcal{N}$-spe\-ci\-fic
Turing machine; $\Theta_{\Gamma,T^\mathrm{MOD[\Gamma]}}$ the characteristic sen\-ten\-ce for the pair
$(\Gamma,T^\mathrm{MOD[\Gamma]})$; $\overline{\Theta}_{\Gamma,T^\mathrm{MOD[\Gamma]}}$ the complementary sentence for
${\Theta}_{\Gamma,T^\mathrm{MOD[\Gamma]}}$.\end{theorem}

\begin{proof} For short, we denote the form~\eqref{eq4}
by
$\Phi_{\Gamma,T^\mathrm{MOD[\Gamma]}}$.

First, we will prove that $\Phi_{\Gamma,T^\mathrm{MOD[\Gamma]}}$ defines a $\mathcal{N}$-complete problem for any pair
$(\Gamma,T^\mathrm{MOD[\Gamma]})$. Let $T^\mathrm{MOD[\Gamma]}$ be a Turing reduction from
$\mathrm{MOD}[\Upsilon_\tau]$ to $\mathrm{MOD}[\Gamma]$. Then, $\Theta_{\Gamma,T^\mathrm{MOD[\Gamma]}}$ is logically
valid. In this case, the \mbox{$\mathcal{L}_\mathcal{N}(\tau)$-}sen\-ten\-ce $\Phi_{\Gamma,T^\mathrm{MOD[\Gamma]}}$ is
logically equivalent to $\Gamma$, and $\Gamma$ defines a $\mathcal{N}$-complete problem. Therefore,
$\Phi_{\Gamma,T^\mathrm{MOD[\Gamma]}}$ defines the same $\mathcal{N}$-complete problem as well. Now, let
$T^\mathrm{MOD[\Gamma]}$ be not a Turing reduction from $\mathrm{MOD}[\Upsilon_\tau]$ to $\mathrm{MOD}[\Gamma]$. Then,
$\mathrm{MOD}[\Theta_{\Gamma,T^\mathrm{MOD[\Gamma]}}]$ is finite, and
$\mathrm{MOD}[\overline\Theta_{\Gamma,T^\mathrm{MOD[\Gamma]}}]$ is cofinite. Therefore, the model class
$\mathrm{MOD}[\Phi_{\Gamma,T^\mathrm{MOD[\Gamma]}}]$ differs from the model class $\mathrm{MOD}[\Upsilon_\tau]$ only in
a finite set of \mbox{$\tau$-}struc\-tu\-res, i.e. $\mathrm{MOD}[\Phi_{\Gamma,T^\mathrm{MOD[\Gamma]}}] \setminus
\mathrm{MOD}[\Upsilon_\tau]$ and $\mathrm{MOD}[\Upsilon_\tau] \setminus
\mathrm{MOD}[\Phi_{\Gamma,T^\mathrm{MOD[\Gamma]}}]$ are both finite. Hence,
$\mathrm{MOD}[\Phi_{\Gamma,T^\mathrm{MOD[\Gamma]}}]$ is \mbox{$\mathcal{N}$-}comp\-le\-te. Consequently,
$\Phi_{\Gamma,T^\mathrm{MOD[\Gamma]}}$ defines a $\mathcal{N}$-complete problem in this case as well.

Second, we will prove that any $\mathcal{N}$-complete problem can be represented by means of the form
$\Phi_{\Gamma,T^\mathrm{MOD[\Gamma]}}$. Let $\Gamma$ be a $\mathcal{L}_\mathcal{N}(\tau)$-sen\-ten\-ce defining an
arbitrary $\mathcal{N}$-complete problem, and $T^\mathrm{MOD[\Gamma]}$ a $\mathcal{N}$-spe\-ci\-fic Turing machine
realizing a Turing reduction from $\mathrm{MOD}[\Upsilon_\tau]$ to $\mathrm{MOD}[\Gamma]$. Then,
$\Theta_{\Gamma,T^\mathrm{MOD[\Gamma]}}$ is logically valid. In this case, the
$\mathcal{L}_\mathcal{N}(\tau)$-sen\-ten\-ce $\Gamma$ is logically equivalent to the
\mbox{$\mathcal{L}_\mathcal{N}(\tau)$-}sen\-ten\-ce $\Phi_{\Gamma,T^\mathrm{MOD[\Gamma]}}$. Since $\Gamma$ defines a
\mbox{$\mathcal{N}$-}co\-m\-ple\-te problem, so does $\Phi_{\Gamma,T^\mathrm{MOD[\Gamma]}}$. This concludes the proof
of the theorem. \end{proof}

Thus, the form \eqref{eq4} can serve as a canonical form providing a syntactic tool for showing
\mbox{$\mathcal{N}$-}comp\-le\-te\-ness: if a problem on unordered structures in $\mathcal{N}$ is defined by a
sen\-ten\-ce of the form, then the problem proves to be \mbox{$\mathcal{N}$-}comp\-le\-te via Turing reductions.

\section{\label{sec4}Logics for completeness}

At first, we will consider ordered structures. Let an arbitrary vocabulary $\tau_{<}$ containing $<$ be given. Let us
address the following question. Can one recursively enumerate all
\mbox{$\mathcal{L}_\mathcal{N}(\tau_{<})$-sen}\-ten\-ces that define \mbox{$\mathcal{N}$-}comp\-le\-te problems on
$\mathrm{STRUC}[\tau_{<}]$, where $\mathcal{N}$ denotes one of the following complexity classes: $\mathrm{NL}$,
$\mathrm{P}$, $\mathrm{coNP}$, $\mathrm{NP}$, or $\mathrm{PSPACE}$? We will answer this question in the negative.

We will use notions of context-free languages~\cite{Ginsburg}, and exploit them in a similar manner as in
\cite{Naidenko}. Let $\Sigma$ be a fixed alphabet containing at least two symbols. By $G$ and $L(G)$ we mean a
context-free grammar and the context-free language defined by this grammar respectively. We assume that there cannot be
any finite \mbox{$\mathcal{N}$-comp}\-le\-te  set. Suppose that $\Phi_1,\Phi_2,\ldots$ is a recursive enumeration of
all \mbox{$\mathcal{L}_\mathcal{N}(\tau_{<})$-sen}\-ten\-ces defining $\mathcal{N}$-comp\-le\-te problems. With every
context-free grammar $G$ with the terminal alphabet $\Sigma$, we associate the following set:
$$
\begin{array}{rl}
  S^{<}_{G}\triangleq\{\:A\in
\mathrm{STRUC}[\tau_{<}] \mid & \!\!\!\!\! (\forall w \in \Sigma^{*})
\textrm{[}\,|w|
> \ell^{(3)}(|\langle A\rangle|)\, \vee\, w\in L(G) \textrm{]}\:\}.
\end{array}
$$

Since the decision problem of $w\in?\,L(G)$ requires at most polynomial time, we have $S^{<}_{G} \in
\mathrm{DTIME}[\log n]$ by analogy with Theorem~\ref{Theorem2}. Then, there is a
\mbox{$\mathrm{FO}(\tau_{<})$-sen}\-ten\-ce $\gamma_G$ that defines $S^{<}_{G}$. Note that if $L(G)=\Sigma^{*}$, then
$\gamma_G$ is logically valid. Otherwise, the set $\mathrm{MOD}[\gamma_G]$ is finite. It follows that the
\mbox{$\mathcal{L}_\mathcal{N}(\tau_{<})$-sen}\-ten\-ce $\gamma_G \wedge \Upsilon_{\tau_{<}}$ defines a
$\mathcal{N}$-comp\-le\-te problem if and only if the statement $L(G)=\Sigma^{*}$ holds.

Let an arbitrary context-free grammar $G$ with the terminal alphabet $\Sigma$ be given. Let us decide whether
$L(G)=\Sigma^{*}$ or not. We concurrently start the following two algorithms. First, we enumerate all strings in
$\Sigma^{*}$. If we can find a string $w$ such that $w\not\in L(G)$, then $L(G)\neq\Sigma^{*}$. Second, we construct
the \mbox{$\mathcal{L}_\mathcal{N}(\tau_{<})$-sen}\-ten\-ce $\gamma_G \wedge \Upsilon_{\tau_{<}}$ for this grammar $G$.
Then, we enumerate the \mbox{$\mathcal{L}_\mathcal{N}(\tau_{<})$-sen}\-ten\-ces $\Phi_1,\Phi_2,\ldots$. If we can find
$\Phi_i$ such that $\Phi_i=\gamma_G \wedge \Upsilon_{\tau_{<}}$, then $L(G)=\Sigma^{*}$. Hence, we can algorithmically
decide whether $L(G)=\Sigma^{*}$ or not. However, this contradicts the fact that the decision problem of $L(G)=?\,
\Sigma^{*}$ is algorithmically undecidable~\cite{Ginsburg}. Thus, one cannot recursively enumerate all
\mbox{$\mathcal{L}_\mathcal{N}(\tau_{<})$-sen}\-ten\-ces that define {$\mathcal{N}$-}comp\-le\-te problems on
$\mathrm{STRUC}[\tau_{<}]$.

Nevertheless, we can use form~\eqref{eq2} in order to recursively enumerate (not all)
\mbox{$\mathcal{L}_\mathcal{N}(\tau_{<})$-sen}\-ten\-ces that define all $\mathcal{N}$-comp\-le\-te problems. Let
$\mathcal{C}_\mathcal{N}(\tau_{<})$ denote the following set of all
\mbox{$\mathcal{L}_\mathcal{N}(\tau_{<})$-sen}\-ten\-ces of form~\eqref{eq2}:
$$
\begin{array}{rl}
\{\Phi\in
\mathcal{L}_\mathcal{N}(\tau_{<})\mid & \!\!\!\!\! \Phi =
(\gamma_{\Gamma,T^\mathrm{MOD[\Gamma]}} \wedge \Gamma) \vee
(\neg\gamma_{\Gamma,T^\mathrm{MOD[\Gamma]}} \wedge
\Upsilon_{\tau_{<}});\ \Gamma\in \mathcal{L}_\mathcal{N}(\tau_{<}),
\\ & \!\!\!\!\!
T^\mathrm{MOD[\Gamma]}\
\textrm{is a \textit{$\mathcal{N}$}-spe\-ci\-fic Turing machine}\}.
\end{array}
$$

Note that $\mathcal{C}_\mathcal{N}(\tau_{<})$ is recursively enumerable since we can effectively enumerate all possible
pairs $(\Gamma,T^\mathrm{MOD[\Gamma]})$. Therefore, the set of all $\mathcal{N}$-complete problems on
$\mathrm{STRUC}[\tau_{<}]$, defined by sentences in $\mathcal{C}_\mathcal{N}(\tau_{<})$, is recursively enumerable as
well.

However, it is hard to determine whether $\mathcal{C}_\mathcal{N}(\tau_{<})$ is recursive. In order to be able to
decide whether $\Phi\in \mathcal{C}_\mathcal{N}(\tau_{<})$, we need to find the corresponding
\mbox{$\mathcal{N}$-}spe\-ci\-fic machine $T^\mathrm{MOD[\Gamma]}$. This seems to be undecidable. Fortunately, we can
change form~\eqref{eq2} in such a way as to obtain a recursive set of
\mbox{$\mathcal{L}_\mathcal{N}(\tau_{<})$-sen}\-ten\-ces defining all $\mathcal{N}$-complete problems on
$\mathrm{STRUC}[\tau_{<}]$.

With each nonempty binary string $w$, we associate the following identically false first order sen\-ten\-ce $\psi_w$:
$$\mathcal{Q}x_1
\ldots\mathcal{Q}x_k[x_1\ne x_1 \wedge \cdots \wedge x_k\ne x_k]$$ where $k=|w|$; $\mathcal{Q}x_i$ denotes either
$\exists x_i$ if $w[i]=1$, or $\forall x_i$ if $w[i]=0$, for every ${1\le i\le k}$. We say that $\psi_w$ encodes the
binary string $w$, and call $\psi_w$ an encoding sentence. Since such a first order sen\-ten\-ce can be considered in a
certain sense as a representation of binary strings, we will use this sen\-ten\-ce for an appropriate encoding of
Turing machines. Let us introduce the following form:
\begin{equation}
(\gamma_{\Gamma,T^\mathrm{MOD[\Gamma]}} \wedge \Gamma) \vee
(\neg\gamma_{\Gamma,T^\mathrm{MOD[\Gamma]}} \wedge \Upsilon_{\tau_{<}}) \vee \psi_{\langle
T^\mathrm{MOD[\Gamma]}\rangle}\label{eq5}
\end{equation}

Note that form~\eqref{eq5} is logically equivalent to form~\eqref{eq2} since $\psi_{\langle
T^\mathrm{MOD[\Gamma]}\rangle}$ is identically false. However, in contrast to~\eqref{eq2}, form~\eqref{eq5} allows us
to easily construct a logic capturing the complexity class of all $\mathcal{N}$-complete problems on all ordered
structures. Indeed, let $\mathcal{CO}\mathcal{L}_\mathcal{N}$ be a mapping from every vocabulary $\tau_{<}$ containing
$<$ to $\mathcal{CO}\mathcal{L}_\mathcal{N}(\tau_{<})$, where $\mathcal{CO}\mathcal{L}_\mathcal{N}(\tau_{<})$ denotes
the following set of $\mathcal{L}_\mathcal{N}(\tau_{<})$-sen\-ten\-ces of form~\eqref{eq5}:
$$
\begin{array}{rl}
\{\Phi\in
\mathcal{L}_\mathcal{N}(\tau_{<})\mid & \!\!\! \Phi =
(\gamma_{\Gamma,T^\mathrm{MOD[\Gamma]}} \wedge \Gamma) \vee
(\neg\gamma_{\Gamma,T^\mathrm{MOD[\Gamma]}} \wedge
\Upsilon_{\tau_{<}}) \vee \psi_{\langle T^\mathrm{MOD[\Gamma]}\rangle};
\\ &  \!\!
\Gamma\in \mathcal{L}_\mathcal{N}(\tau_{<}),\
T^\mathrm{MOD[\Gamma]}\ \textrm{is a
\textit{$\mathcal{N}$}-spe\-ci\-fic Turing machine}
\}.
\end{array}
$$

 Then, the following theorem holds.

\begin{theorem}
\label{Theorem8} Let
 $\mathcal{L}_\mathcal{N}$ be a logic capturing a
complexity class $\mathcal{N}$ (among the classes $\mathrm{NL}$, $\mathrm{P}$, $\mathrm{coNP}$, $\mathrm{NP}$, and
$\mathrm{PSPACE}$) on $\mathrm{STRUC}[\tau_{<}]$ for some fixed vocabulary $\tau_{<}$ containing $<$. Then,
$\mathcal{CO}\mathcal{L}_\mathcal{N}$ is a decidable fragment of $\mathcal{L}_\mathcal{N}$ that captures the complexity
class of all $\mathcal{N}$-complete problems on $\mathrm{STRUC}[\tau_{<}]$.\end{theorem}

\begin{proof} Let us show how to effectively decide whether $\Phi\in \mathcal{CO}\mathcal{L}_\mathcal{N}(\tau_{<})$
or not, given an arbitrary $\Phi\in \mathcal{L}_\mathcal{N}(\tau_{<})$. At first, we verify whether $\Phi$ is of the
form $(\gamma \wedge \Gamma) \vee (\neg\gamma \wedge \Upsilon) \vee \psi$ or not, where $\Gamma$ and $\Upsilon$ are
both $\mathcal{L}_\mathcal{N}(\tau_{<})$-sen\-ten\-ces; $\gamma$ and $\psi$ are both first order sen\-ten\-ces. If this
is not the case, then $\Phi\not\in \mathcal{CO}\mathcal{L}_\mathcal{N}$. Otherwise, we apply the operator
${T_{\tau_{<}}}$ to $\Upsilon(R(x))$, and obtain $\Upsilon_{\tau_{<}}$.
If $\Upsilon\neq \Upsilon_{\tau_{<}}$, then $\Phi\not\in \mathcal{CO}\mathcal{L}_\mathcal{N}$. Otherwise, we check
whether $\psi$ is an encoding sen\-ten\-ce $\psi_w$ or not. If this is not the case, then $\Phi\not\in
\mathcal{CO}\mathcal{L}_\mathcal{N}$. Otherwise, we verify whether $w$ is a code of some
\mbox{$\mathcal{N}$-spe}\-ci\-fic Turing machine $T^\mathrm{MOD[\Gamma]}$ or not. If $w$ does not encode any
\mbox{$\mathcal{N}$-spe}\-ci\-fic Turing machine $T^\mathrm{MOD[\Gamma]}$, then $\Phi\not\in
\mathcal{CO}\mathcal{L}_\mathcal{N}$. Otherwise, we recover $T^\mathrm{MOD[\Gamma]}$ from its code $w$. Then, we
construct the characteristic
sen\-ten\-ce $\gamma_{\Gamma,T^\mathrm{MOD[\Gamma]}}$ for the pair $(\Gamma,T^\mathrm{MOD[\Gamma]})$. If
$\gamma=\gamma_{\Gamma,T^\mathrm{MOD[\Gamma]}}$, then $\Phi\in \mathcal{CO}\mathcal{L}_\mathcal{N}$. Otherwise,
$\Phi\not\in \mathcal{CO}\mathcal{L}_\mathcal{N}$.

Thus, the set $\mathcal{CO}\mathcal{L}_\mathcal{N}(\tau_{<})$ of $\mathcal{L}_\mathcal{N}(\tau_{<})$-sen\-ten\-ces is
recursive, and $\mathcal{CO}\mathcal{L}_\mathcal{N}$ is a logic that represents a decidable fragment of
$\mathcal{L}_\mathcal{N}$. Since form~\eqref{eq5} is logically equivalent to form~\eqref{eq2}, Theorem~\ref{Theorem4}
holds for form~\eqref{eq5} as well as for form~\eqref{eq2}. Then, a problem $\Pi\subseteq \mathrm{STRUC}[\tau_{<}]$ is
\mbox{$\mathcal{N}$-}comp\-le\-te if and only if there exists a
\mbox{$\mathcal{CO}\mathcal{L}_\mathcal{N}(\tau_{<})$-sen}\-ten\-ce $\Lambda$ defining $\Pi$. This concludes the proof
of the theorem. \end{proof}

\begin{corollary}
\label{corollary9} If a logic $\mathcal{L}_\mathcal{N}$ captures a complexity class $\mathcal{N}$ on all ordered
structures, then the logic $\mathcal{CO}\mathcal{L}_\mathcal{N}$ captures the complexity class of all
$\mathcal{N}$-complete problems on all ordered structures.\end{corollary}

\begin{proof} It is immediate from the theorem. \end{proof}

Now, we proceed to unordered non-Aristotelian structures.
Let us introduce the following form:
\begin{equation}
(\Theta_{\Gamma,T^\mathrm{MOD[\Gamma]}} \wedge \Gamma) \vee
(\overline{\Theta}_{\Gamma,T^\mathrm{MOD[\Gamma]}} \wedge \Upsilon_\tau) \vee \psi_{\langle
T^\mathrm{MOD[\Gamma]}\rangle}\label{eq6}
\end{equation}

Note that form~\eqref{eq6} is logically equivalent to form~\eqref{eq4} since $\psi_{\langle
T^\mathrm{MOD[\Gamma]}\rangle}$ is identically false. However, in contrast to~\eqref{eq4}, form~\eqref{eq6} allows us
to easily construct a logic capturing the complexity class of all $\mathcal{N}$-complete problems on all
non-Aristotelian structures. Indeed, let $\mathcal{C}\mathcal{L}_\mathcal{N}$ be a mapping from every non-Aristotelian
vocabulary $\tau$ to $\mathcal{C}\mathcal{L}_\mathcal{N}(\tau)$, where $\mathcal{C}\mathcal{L}_\mathcal{N}(\tau)$
denotes the following set of $\mathcal{L}_\mathcal{N}(\tau)$-sen\-ten\-ces of form~\eqref{eq6}:
$$
\begin{array}{rl}
\{\Phi\in
\mathcal{L}_\mathcal{N}(\tau)\mid & \!\!\!\!\! \Phi =
(\Theta_{\Gamma,T^\mathrm{MOD[\Gamma]}} \wedge \Gamma) \vee
(\overline{\Theta}_{\Gamma,T^\mathrm{MOD[\Gamma]}} \wedge
\Upsilon_\tau) \vee \psi_{\langle T^\mathrm{MOD[\Gamma]}\rangle};\\ & \!\!\!\!
\Gamma\in \mathcal{L}_\mathcal{N}(\tau),\
T^\mathrm{MOD[\Gamma]}\ \textrm{is a
\textit{$\mathcal{N}$}-spe\-ci\-fic Turing machine}
\}.
\end{array}
$$
Then, the following theorem holds.

\begin{theorem}
\label{Theorem10} Let $\mathcal{L}_\mathcal{N}$ be a logic capturing a complexity class $\mathcal{N}$ (among the
classes $\mathrm{coNP}$, $\mathrm{NP}$, and $\mathrm{PSPACE}$) on $\mathrm{STRUC}[\tau]$ for a fixed arbitrary
non-Aristotelian vocabulary $\tau$. Then, $\mathcal{C}\mathcal{L}_\mathcal{N}$ is a decidable fragment of
$\mathcal{L}_\mathcal{N}$ that captures the complexity class of all $\mathcal{N}$-complete problems on
$\mathrm{STRUC}[\tau]$.\end{theorem}

\begin{proof} Let us show how to effectively decide whether $\Phi\in \mathcal{C}\mathcal{L}_\mathcal{N}(\tau)$ or
not, given an arbitrary $\Phi\in \mathcal{L}_\mathcal{N}(\tau)$. At first, we verify whether $\Phi$ is of the form
$(\Theta \wedge \Gamma) \vee (\Xi \wedge \Upsilon) \vee \psi$ or not, where $\Theta$, $\Gamma$, $\Xi$, and $\Upsilon$
are $\mathcal{L}_\mathcal{N}(\tau)$-sen\-ten\-ces; $\psi$ is a first order sen\-ten\-ce. If this is not the case, then
$\Phi\not\in \mathcal{C}\mathcal{L}_\mathcal{N}$. Otherwise, we apply the operator ${T_{\tau}}$ to $\Upsilon(R(x,y))$,
and obtain $\Upsilon_\tau$.
If $\Upsilon\neq \Upsilon_\tau$, then $\Phi\not\in
\mathcal{C}\mathcal{L}_\mathcal{N}$. Otherwise, we check whether $\psi$ is an encoding sen\-ten\-ce $\psi_w$ or not. If
this is not the case, then $\Phi\not\in \mathcal{C}\mathcal{L}_\mathcal{N}$. Otherwise, we verify whether $w$ is a code
of some \mbox{$\mathcal{N}$-spe}\-ci\-fic Turing machine $T^\mathrm{MOD[\Gamma]}$ or not. If $w$ does not encode any
\mbox{$\mathcal{N}$-spe}\-ci\-fic Turing machine $T^\mathrm{MOD[\Gamma]}$, then $\Phi\not\in
\mathcal{C}\mathcal{L}_\mathcal{N}$. Otherwise, we recover $T^\mathrm{MOD[\Gamma]}$ from its code $w$. Then, we
construct the characteristic
sen\-ten\-ce $\Theta_{\Gamma,T^\mathrm{MOD[\Gamma]}}$ for the pair $(\Gamma,T^\mathrm{MOD[\Gamma]})$. Also, we
construct the complementary sentence $\overline{\Theta}_{\Gamma,T^\mathrm{MOD[\Gamma]}}$ for
$\Theta_{\Gamma,T^\mathrm{MOD[\Gamma]}}$, given $T^\mathrm{MOD[\Gamma]}$ and $\Gamma$. If
$\Theta=\Theta_{\Gamma,T^\mathrm{MOD[\Gamma]}}$ and $\Xi=\overline{\Theta}_{\Gamma,T^\mathrm{MOD[\Gamma]}}$, then
$\Phi\in \mathcal{C}\mathcal{L}_\mathcal{N}$. Otherwise, $\Phi\not\in \mathcal{C}\mathcal{L}_\mathcal{N}$.

Thus, the set $\mathcal{C}\mathcal{L}_\mathcal{N}(\tau)$ of $\mathcal{L}_\mathcal{N}(\tau)$-sen\-ten\-ces is recursive,
and $\mathcal{C}\mathcal{L}_\mathcal{N}$ is a logic that represents a decidable fragment of $\mathcal{L}_\mathcal{N}$.
Since form~\eqref{eq6} is logically equivalent to form~\eqref{eq4}, Theorem~\ref{Theorem7} holds for form~\eqref{eq6}
as well as for form~\eqref{eq4}. Then, a problem $\Pi\subseteq \mathrm{STRUC}[\tau]$ is \mbox{$\mathcal{N}$-}complete
if and only if there exists a $\mathcal{C}\mathcal{L}_\mathcal{N}(\tau)$-sen\-ten\-ce $\Lambda$ defining ${\Pi}$. This
concludes the proof of the theorem. \end{proof}

\begin{corollary}
\label{corollary11} If a logic $\mathcal{L}_\mathcal{N}$ captures a complexity class $\mathcal{N}$ on all structures,
then the logic $\mathcal{C}\mathcal{L}_\mathcal{N}$ captures the complexity class of all $\mathcal{N}$-complete
problems on all non-Aristotelian structures.\end{corollary}

\begin{proof} It is immediate from the theorem. \end{proof}

\section{Canonical form and logic for \texorpdfstring{$\mathbf{NP\cap coNP}$}{NP n coNP}}

Let a vocabulary $\tau_{<}$ containing $<$ be given. In a very similar way (see section \ref{sec4}), we will show that
one cannot recursively enumerate all $\mathrm{SO}\exists(\tau_{<})$-sen\-ten\-ces defining problems in $\mathrm{NP\cap
coNP}$, provided that $\mathrm{NP\cap coNP\neq NP}$. Let $\Gamma$ denote a distinguished
\mbox{$\mathrm{SO}\exists(\tau_{<})$-sen}\-ten\-ce that defines some problem in $\mathrm{NP\cap coNP}$. Here the
notation $\Upsilon_{\tau_{<}}$ stands for a \mbox{$\mathrm{SO}\exists(\tau_{<})$-sen}\-ten\-ce defining a
\mbox{$\mathrm{NP}$-}co\-m\-p\-le\-te problem. Let us consider a \mbox{$\mathrm{SO}\exists(\tau_{<})$-sen}\-ten\-ce
$\Phi$ of the form $(\gamma_G \wedge \Gamma) \vee (\neg{\gamma}_G \wedge \Upsilon_{\tau_{<}})$.
Note that if the statement $L(G)=\Sigma^{*}$ holds, then $\gamma_G$ is logically valid. It follows that $\Phi$ is
logically equivalent to $\Gamma$, and $\Phi$ defines a problem in $\mathrm{NP\cap coNP}$ in this case. If the statement
$L(G)\neq\Sigma^{*}$ holds, then $\mathrm{MOD}[\gamma_G]$ is finite. Therefore, the set
$\mathrm{MOD}[\Upsilon_{\tau_{<}}]$ differs from the set $\mathrm{MOD}[\Phi]$ only in a finite set of
\mbox{$\tau_{<}$-}struc\-tu\-res, i.e. $\mathrm{MOD}[\Upsilon_{\tau_{<}}] \setminus \mathrm{MOD}[\Phi]$ and
$\mathrm{MOD}[\Phi] \setminus \mathrm{MOD}[\Upsilon_{\tau_{<}}]$ are both finite. Then, $\Phi$ defines a
$\mathrm{NP}$-complete problem that cannot be in $\mathrm{NP\cap coNP}$, on the assumption $\mathrm{NP\cap coNP\neq
NP}$. Thus, $\Phi$ defines a problem in $\mathrm{NP\cap coNP}$ if and only if the statement $L(G)=\Sigma^{*}$ holds.

Let an arbitrary context-free grammar $G$ with the terminal alphabet $\Sigma$ be given. We construct the
\mbox{$\mathrm{SO}\exists(\tau_{<})$-sen}\-ten\-ce $(\gamma_G \wedge \Gamma) \vee (\neg{\gamma}_G \wedge
\Upsilon_{\tau_{<}})$ for this grammar $G$. Suppose that we recursively enumerate all
\mbox{$\mathrm{SO}\exists(\tau_{<})$-sen}\-ten\-ces $\Phi_1,\Phi_2,\ldots$ that define problems in $\mathrm{NP\cap
coNP}$. If we find $\Phi_i$ such that $\Phi_i=(\gamma_G \wedge \Gamma) \vee (\neg{\gamma}_G \wedge
\Upsilon_{\tau_{<}})$, then we have $L(G)=\Sigma^{*}$. However, this contradicts the fact that the statement
$L(G)=\Sigma^{*}$ is algorithmically undecidable (see section \ref{sec4}). Thus, one cannot recursively enumerate all
\mbox{$\mathrm{SO}\exists(\tau_{<})$-sen}\-ten\-ces that define problems on $\mathrm{STRUC}[\tau_{<}]$
in~$\mathrm{NP\cap coNP}$, provided that $\mathrm{NP\cap coNP\neq NP}$.

Then, we state the following question. What is it about a $\mathrm{SO}\exists$-sentence that makes its defined problem
be in $\mathrm{NP\cap coNP}$? We will answer this question in the affirmative.

Gurevich \cite{Gurevich} conjectured that no logic captures the complexity class $\mathrm{NP\cap coNP}$. The conjecture
is based on the fact that the existence of such a logic (in the sense of Gurevich) for $\mathrm{NP\cap coNP}$ implies
the existence of a complete problem in $\mathrm{NP\cap coNP}$. However, there does not exist such a complete problem
relative to a certain oracle \cite{Sipser}. Therefore, $\mathrm{NP\cap coNP}$ very likely has no complete problem.

Nevertheless, it is not necessary for a logic to imply the existence of a complete problem. For example, no complete
problem is known in the complexity class $\mathrm{PH}$, whereas second order logic captures $\mathrm{PH}$ (see
\cite{Stockmeyer}).

We can extend our approach beyond complete problems. A canonical form (similar to form~\eqref{eq6}) will be developed
to define problems in $\mathrm{NP\cap coNP}$. Moreover, we will develop a logic that captures $\mathrm{NP\cap coNP}$ and
does not require the existence of any complete problem in $\mathrm{NP\cap coNP}$.

Let $\tau$ be a fixed vocabulary, where $\tau$ may be any (including Aristotelian) vocabulary. Then, with each pair
$(\Lambda,\Gamma)$ of \mbox{$\mathrm{SO}\exists(\tau)$-}sen\-ten\-ces, we associate the following set of
\mbox{$\tau$-}struc\-tu\-res:
$$
\begin{array}{rl}
  S_{\Lambda,\Gamma}\triangleq\{\:A\in
\mathrm{STRUC}[\tau] \mid    & \!\!\!\!\! (\forall B \in
\mathrm{STRUC}[\tau]) \textrm{[}\: \|B\|
> \ell^{(2)}(|\langle A\rangle|)\, \vee\\
& \!\!\!\!\! (B \models \Lambda \Leftrightarrow B \not\models
\Gamma)\:\textrm{]}\:\}.
\end{array}
$$

Note that the set $S_{\Lambda,\Gamma}$ has very similar properties as the previously considered set
$S_{\Gamma,T^\mathrm{MOD[\Gamma]}}$. We have the following theorem.

\begin{theorem}
\label{Theorem12} $S_{\Lambda,\Gamma} \in \mathrm{DTIME}[n]$.\end{theorem}

\begin{proof} It is similar to the proof of Theorem~\ref{Theorem5} in many respects. Let us decide whether $A\in S_{\Lambda,\Gamma}$
or not, given a \mbox{$\tau$-}struc\-tu\-re $A$. Let $\hat{n}$ denote $|\langle A\rangle|$.

First, we need to compute the number $\ell^{(2)}(\hat{n})$.
It is obvious that we can find $\ell^{(2)}(\hat{n})$ in time $O(\hat{n})$.

Second, we enumerate all \mbox{$\tau$-}structures $B\in \mathrm{STRUC}[\tau]$ such that $\|B\|\le\ell^{(2)}(\hat{n})$.
Since $|\langle B\rangle|\le p(\|B\|)$, the value $|\langle B\rangle|$ is bounded above by $p(\log \log  \hat{n})$,
where $p$ is a polynomial dependent on the vocabulary $\tau$.  Therefore, this enumeration takes time $O(2^{p(\log \log
\hat{n})})$. Then, for each \mbox{$\tau$-}structure $B$ of the enumeration, we need to verify the following condition:
\begin{equation}B \models \Lambda
\Leftrightarrow B \not\models \Gamma.\label{eq7}
\end{equation}

If condition~\eqref{eq7} is not satisfied for some \mbox{$\tau$-}struc\-tu\-re $B$ of the enumeration, then the input
\mbox{$\tau$-}struc\-tu\-re $A$ does not belong to $S_{\Lambda,\Gamma}$. Otherwise, $A \in S_{\Lambda,\Gamma}$. Now,
let us estimate the running time for verifying condition~\eqref{eq7}.

We can verify \eqref{eq7} for one \mbox{$\tau$-}struc\-tu\-re $B$ in time $O(2^{p'(\log \log \hat{n})})$, where $p'$ is
a polynomial. This follows from the fact that any model-checking problem in $\mathrm{NP\cap coNP}$ is decidable at most
in exponential time.

Consequently, the total time to verify~\eqref{eq7} for all \mbox{$\tau$-}struc\-tu\-res $B$ of the enumeration is
$O(2^{p(\log \log \hat{n})}\cdot2^{p'(\log \log  \hat{n})})$, or, in short, $O(2^{p''(\log \log  \hat{n})})$, where
$p''=p+p'$. Since $\displaystyle\lim_{\hat{n}\to\infty}\frac{2^{p''(\log \log \hat{n})}}{\hat{n}}=0$ for any polynomial
$p''$,
this
time can be approximately estimated at most as $O(\hat{n})$.

Thus, the overall running time consists of the time for computing $\ell^{(2)}(\hat{n})$ and of the time for
verifying~\eqref{eq7} on all \mbox{$\tau$-}struc\-tu\-res $B$ of the enumeration. The both of these times are estimated
as $O(\hat{n})$. Therefore, so is the overall running time. This concludes the proof of the theorem. \end{proof}

\begin{corollary}
\label{corollary13} Given a pair $(\Lambda,\Gamma)$ of $\mathrm{SO}\exists(\tau)$-sen\-ten\-ces, one can construct a
$\mathrm{SO}\exists(\tau)$-sen\-ten\-ce that defines $S_{\Lambda,\Gamma}$.\end{corollary}

\begin{proof} It follows from the statement $S_{\Lambda,\Gamma} \in \mathrm{DTIME}[n]$. \end{proof}

Then, by $\Theta_{\Lambda,\Gamma}$ we denote the constructed \mbox{$\mathrm{SO}\exists(\tau)$-}sen\-ten\-ce defining
the model class $S_{\Lambda,\Gamma}$. This $\Theta_{\Lambda,\Gamma}$ characterizes $\Lambda$ and $\Gamma$ in the following way:
if $\Theta_{\Lambda,\Gamma}$ is logically valid, then $\mathrm{MOD}[\Lambda]$ and $\mathrm{MOD}[\Gamma]$ are both in
$\mathrm{NP\cap coNP}$. We call the \mbox{$\mathrm{SO}\exists(\tau)$-}sen\-ten\-ce $\Theta_{\Lambda,\Gamma}$ the
characteristic sen\-ten\-ce for the pair $(\Lambda,\Gamma)$. Then, the following theorem holds.

\begin{theorem}
\label{Theorem14} Let $\Pi\subseteq \mathrm{STRUC}[\tau]$ be a model class. Then, $\Pi$ is in $\mathrm{NP\cap coNP}$ if
and only if there exists a $\mathrm{SO}\exists(\tau)$-sen\-ten\-ce $\Omega$ such that ${\mathrm{MOD}[\Omega]=\Pi}$ and
$\Omega$ is of the form
\begin{equation}
(\Theta_{\Lambda,\Gamma} \wedge \Gamma) \vee
\psi_{\langle\Lambda\rangle}\label{eq8}
\end{equation}
where $\Lambda$ and $\Gamma$ are both $\mathrm{SO}\exists(\tau)$-sen\-ten\-ces; $\Theta_{\Lambda,\Gamma}$ the
characteristic sen\-ten\-ce for $(\Lambda,\Gamma)$; $\psi_{\langle\Lambda\rangle}$ the sen\-ten\-ce encoding the binary
string $\langle\Lambda\rangle$.\end{theorem}

\begin{proof} For short, we denote form~\eqref{eq8} by $\Phi_{\Lambda,\Gamma}$.

First, we will prove that $\Phi_{\Lambda,\Gamma}$ defines a problem in $\mathrm{NP\cap coNP}$ for any pair
$(\Lambda,\Gamma)$ of \mbox{$\mathrm{SO}\exists(\tau)$-}sen\-ten\-ces. Suppose that
$\mathrm{MOD}[\Gamma]=\mathrm{STRUC}[\tau]\setminus\mathrm{MOD}[\Lambda]$. Then, $\mathrm{MOD}[\Gamma]\in
\mathrm{coNP}$, and, therefore, we have $\mathrm{MOD}[\Gamma]\in \mathrm{NP\cap coNP}$. Note that
$\Theta_{\Lambda,\Gamma}$ is logically valid in this case. Then, $\Phi_{\Lambda,\Gamma}$ is logically equivalent to
$\Gamma$. It follows that $\Phi_{\Lambda,\Gamma}$ defines the problem $\mathrm{MOD}[\Gamma]$ in $\mathrm{NP\cap coNP}$.
Now, suppose that $\mathrm{MOD}[\Gamma]\neq\mathrm{STRUC}[\tau]\setminus\mathrm{MOD}[\Lambda]$. Then,
$\mathrm{MOD}[\Theta_{\Lambda,\Gamma}]$ is finite, and so is $\mathrm{MOD}[\Phi_{\Lambda,\Gamma}]$. Since the set
$\mathrm{MOD}[\Phi_{\Lambda,\Gamma}]$ is finite, $\mathrm{MOD}[\Phi_{\Lambda,\Gamma}]\in \mathrm{NP\cap coNP}$ in this
case as well.

Second, we will prove that any problem in $\mathrm{NP\cap coNP}$ can be defined by means of the form
$\Phi_{\Lambda,\Gamma}$. Let $\Gamma$ be a $\mathrm{SO}\exists(\tau)$-sen\-ten\-ce defining an arbitrary problem in
$\mathrm{NP\cap coNP}$. Then, $\mathrm{MOD}[\Gamma]$ and $\mathrm{STRUC}[\tau]\setminus\mathrm{MOD}[\Gamma]$ are both
in $\mathrm{NP}$. Therefore, there exists a $\mathrm{SO}\exists(\tau)$-sen\-ten\-ce $\Lambda$ that defines the model
class $\mathrm{STRUC}[\tau]\setminus\mathrm{MOD}[\Gamma]$. Note that the characteristic sentence
$\Theta_{\Lambda,\Gamma}$ for the pair $(\Lambda, \Gamma)$ is logically valid in this case. Moreover, $\Gamma$ is
logically equivalent to $\Phi_{\Lambda,\Gamma}$. Since $\Gamma$ defines a problem in $\mathrm{NP\cap coNP}$, so does
$\Phi_{\Lambda,\Gamma}$. This concludes the proof of the theorem.
\end{proof}

Thus, form~\eqref{eq8} can serve as a canonical form for $\mathrm{NP\cap coNP}$. Now, we can use form~\eqref{eq8} in
order to develop a logic capturing $\mathrm{NP\cap coNP}$. By $\mathcal{L}_\mathrm{NP\cap coNP}$ we mean a mapping from
every vocabulary $\tau$ to $\mathcal{L}_\mathrm{NP\cap coNP}(\tau)$, where $\mathcal{L}_\mathrm{NP\cap coNP}(\tau)$
denotes the following set of $\mathrm{SO}\exists(\tau)$-sen\-ten\-ces:
$$
\begin{array}{rl}
\{\Phi\in \mathrm{SO}\exists(\tau)\mid & \!\!\!\!\!
\Phi = (\Theta_{\Lambda,\Gamma} \wedge \Gamma) \vee
\psi_{\langle\Lambda\rangle}; \ \Gamma\in
\mathrm{SO}\exists(\tau), \Lambda\in \mathrm{SO}\exists(\tau)
\}.
\end{array}
$$
Then, the following theorem holds.

\begin{theorem}
\label{Theorem15} $\mathcal{L}_\mathrm{NP\cap coNP}$ is a decidable fragment of $\mathrm{SO}\exists$ logic that
captures $\mathrm{NP\cap coNP}$.\end{theorem}

\begin{proof} It is similar to the proof of Theorem~\ref{Theorem10}. Let us fix an arbitrary vocabulary $\tau$. We will show how to
effectively decide whether $\Phi\in \mathcal{L}_\mathrm{NP\cap coNP}(\tau)$ or not, given an arbitrary $\Phi\in
\mathrm{SO}\exists(\tau)$. At first, we verify whether $\Phi$ is of the form $(\Theta \wedge \Gamma) \vee \psi$ or not,
where $\Gamma$ and $\Theta$ are both $\mathrm{SO}\exists(\tau)$-sen\-ten\-ces;  $\psi$ is a first order sen\-ten\-ce.
If this is not the case, then $\Phi\not\in \mathcal{L}_\mathrm{NP\cap coNP}$. Otherwise, we check whether $\psi$ is an
encoding sen\-ten\-ce $\psi_w$ or not. If this is not the case, then $\Phi\not\in \mathcal{L}_\mathrm{NP\cap coNP}$.
Otherwise, we check whether $w$ encodes some well-formed \mbox{$\mathrm{SO}\exists(\tau)$-}sen\-ten\-ce $\Lambda$ or
not. If this is not the case, then $\Phi\not\in \mathcal{L}_\mathrm{NP\cap coNP}$. Otherwise, we recover $\Lambda$ from
its code $w$. Then, we construct the characteristic
sen\-ten\-ce $\Theta_{\Lambda,\Gamma}$ for the pair $(\Lambda, \Gamma)$. If $\Theta=\Theta_{\Lambda,\Gamma}$, then
$\Phi\in \mathcal{L}_\mathrm{NP\cap coNP}$. Otherwise, $\Phi\not\in \mathcal{L}_\mathrm{NP\cap coNP}$.

Thus, the set $\mathcal{L}_\mathrm{NP\cap coNP}(\tau)$ of $\mathrm{SO}\exists(\tau)$-sen\-ten\-ces is recursive, and
$\mathcal{L}_\mathrm{NP\cap coNP}$ is a logic that represents a decidable fragment of $\mathrm{SO}\exists$. Note that
Theorem~\ref{Theorem14} implies that a problem $\Pi\subseteq \mathrm{STRUC}[\tau]$ is in $\mathrm{NP\cap coNP}$ if and
only if there exists a $\mathcal{L}_\mathrm{NP\cap coNP}(\tau)$-sen\-ten\-ce $\Omega$ defining $\Pi$. This concludes
the proof of the theorem. \end{proof}

\section{Concluding remarks}

We have developed canonical forms for problems that are complete via Turing reductions. Also, we have shown that any
complete problem
can be easily defined by one of these forms. Besides, we
have provided an evidence that there cannot be any complete problem on Aristotelian structures in the complexity
classes $\mathrm{P}$, $\mathrm{coNP}$, $\mathrm{NP}$, and $\mathrm{PSPACE}$.

Logics for complete problems in the complexity classes $\mathrm{NL}$, $\mathrm{P}$, $\mathrm{coNP}$, $\mathrm{NP}$, and
$\mathrm{PSPACE}$ have been developed on the basis of the canonical form~\eqref{eq5} that defines these problems on
ordered structures. On the other hand, logics for complete problems in the complexity classes $\mathrm{coNP}$,
$\mathrm{NP}$, and $\mathrm{PSPACE}$ have been developed on the basis of the canonical form~\eqref{eq6} that defines
these problems on unordered non-Aristotelian structures. It is very likely that analogous canonical forms can be also
developed to construct logics for complete problems in other complexity classes.

Besides, we have extended our approach beyond complete problems. Using a similar form, we have developed a logic that
captures the complexity class $\mathrm{NP\cap coNP}$ which very likely contains no complete problem. Note that a
recursive enumeration of all problems in $\mathrm{NP\cap coNP}$ was considered to be difficult (see, for instance,
\cite{Dawar,Papadimitriou}). Up to our knowledge, no such recursive enumeration was known till now. Dawar \cite{Dawar}
pointed out: ``...the natural set of witnesses for $\mathrm{NP\cap coNP}$ is not recursively enumerable. Thus, finding
a recursively enumerable set of witnesses would require a fundamentally new characterization of the class and would be
a major breakthrough in complexity theory.''. Nevertheless, using the logic $\mathcal{L}_\mathrm{NP\cap coNP}$ to
capture $\mathrm{NP\cap coNP}$, we recursively enumerate all problems in $\mathrm{NP\cap coNP}$ by enumerating the
sentences of this logic.

\begin{figure}[htbp]
\centering
\ifpdf
\includegraphics[width=\textwidth]{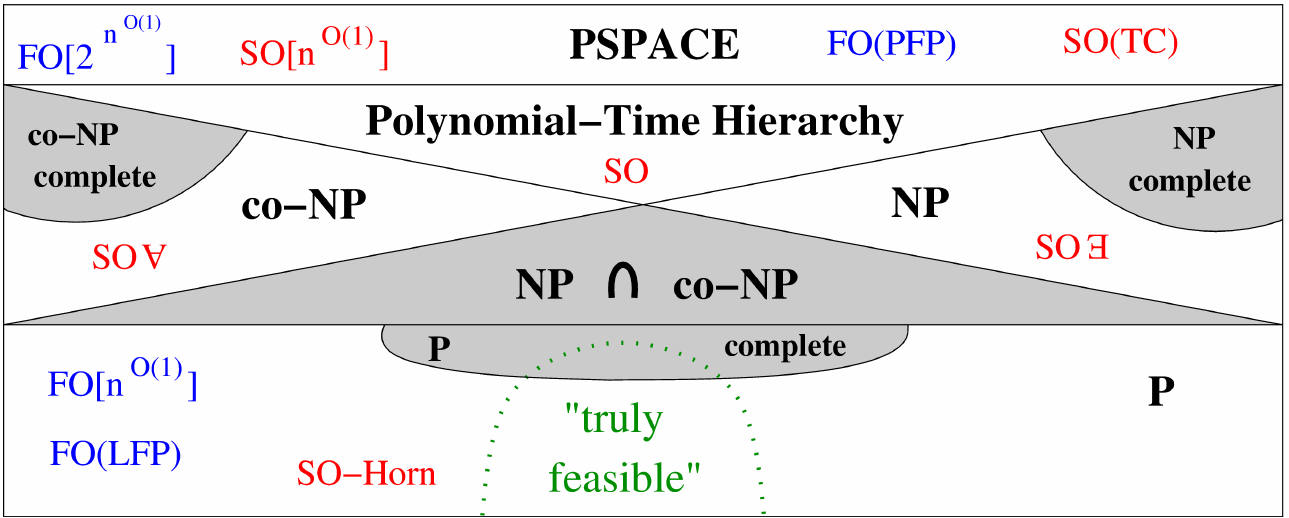}
\else
\includegraphics[width=\textwidth]{pic2.eps}
\fi
\caption[]{The World of Descriptive and Computational Complexity from $\mathrm{P}$ to \mbox{$\mathrm{PSPACE}$}:
shaded areas indicate our developed logics for \mbox{$\mathrm{NP}$-}comp\-le\-te problems, \mbox{$\mathrm{coNP}$-}comp\-le\-te problems,
\mbox{$\mathrm{P}$-}comp\-le\-te problems, and $\mathrm{NP}\cap \mathrm{coNP}$.}
\label{pic2}
\end{figure}

In conclusion, we have modified a fragment of Immerman's diagram \cite{Immerman} in respect to the complexity classes
from $\mathrm{P}$ to $\mathrm{PSPACE}$, as shown in Figure~\ref{pic2} (cf. Figure~\ref{pic1}). For purposes of clarity, in the diagram we
have permitted ourself to shade areas depicting the following complexity classes: \mbox{$\mathrm{NP}$-}complete problems,
\mbox{$\mathrm{coNP}$-}complete problems, \mbox{$\mathrm{P}$-}complete problems, and ${\mathrm{NP\cap coNP}}$ for which we have developed
logics for the first time. Moreover, for the complexity classes of \mbox{$\mathrm{PSPACE}$-}comp\-le\-te problems and
\mbox{$\mathrm{NL}$-}comp\-le\-te problems which are not depicted in Immerman's diagram, we have developed logics as well.

{%

}
\end{sloppypar}
\end{document}